\documentclass[11pt]{article}

\usepackage{makeidx}
\usepackage{color}
\usepackage{graphicx}
\usepackage{hyperref}

\usepackage[makeroom]{cancel}
\usepackage[toc,page]{appendix}
\usepackage{etex,etoolbox}
\usepackage{environ}
\usepackage{titlesec}
\usepackage{etoolbox, chngcntr}
\AtBeginEnvironment{appendices}{%
 \titleformat{\section}{\bfseries\Large}{\appendixname~\thesection:}{0.5em}{}%
 \titleformat{\subsection}{\bfseries\large}{\thesubsection}{0.5em}{}%
\counterwithin{equation}{section}
}
\usepackage [
n,
advantage,
operators,
sets,
adversary,
landau,
probability,
notions,
logic,
ff,
mm,
primitives,
events,
complexity,
asymptotics,
keys
] {cryptocode}

\usepackage{mathtools}
\usepackage{amssymb, amsthm, amsmath}
\usepackage{thmtools} 
\usepackage{braket}
\usepackage{MyMnsymbol}
\usepackage{algorithm, algorithmicx, algpseudocode}
\floatname{algorithm}{Protocol}
\usepackage{float}
\usepackage{enumitem}
\usepackage[textsize=small]{todonotes}
\usepackage[makeroom]{cancel}
\def\EQ#1{\begin{eqnarray}#1\end{eqnarray}} 

\usepackage{booktabs}
\usepackage{multirow}
\usepackage{subcaption}

\usepackage{breakcites}
\allowdisplaybreaks 

\newcommand*{\cA}{\mathcal{A}}

\newcommand*{\cD}{\mathcal{D}}

\newcommand*{\cF}{\mathcal{F}}

\newcommand*{\cK}{\mathcal{K}}

\newcommand*{\cR}{\mathcal{R}}

\newcommand*{\cX}{\mathcal{X}}
\newcommand*{\cY}{\mathcal{Y}}
\newcommand*{\cZ}{\mathcal{Z}}

\newcommand*{\E}{\mathbb{E}}
\newcommand*{\Z}{\mathbb{Z}}
\newcommand*{\N}{\mathbb{N}}

\newcommand{\LWE}{\textsc{LWE}}
\newcommand{\GapSVP}{\textsc{GapSVP}}
\newcommand{\SIVP}{\textsc{SIVP}}

\newcommand*{\accept}{\mathtt{accept}}
\newcommand*{\abort}{\mathtt{abort}}
\newcommand*{\gadxor}{\mathtt{Gad}_{\oplus}}
\newcommand*{\quin}{\mathtt{in}}
\newcommand*{\quout}{\mathtt{out}}
\newcommand{\pr}[1]{\Pr \left[ #1 \right]}

\newcommand*{\piAFourWeak}{\pi_{A_4}}
\newcommand*{\piBFourWeak}{\pi_{B_4}}

\newcommand*{\piAFourXOR}{\pi_{A_4\xor}}

\newcommand*{\piAFourXORchunk}{\pi_{A_4\xor c}}
\newcommand*{\piBFourXORchunk}{\pi_{B_4\xor c}}

\newcommand*{\eps}{\varepsilon}

\makeatletter
\providecommand{\@fourthoffour}[4]{#4}

\newcounter{counttheorems}

\newcommand\fixstatement[2][\proofname\space of]{%
  \ifcsname thmt@original@#2\endcsname
    \AtEndEnvironment{#2}{%
      \xdef\pat@uniqlabel{\thecounttheorems}%
      \xdef\pat@label{\expandafter\expandafter\expandafter
        \@fourthoffour\csname thmt@original@#2\endcsname\space\@currentlabel}%
      \xdef\pat@proofof{\@nameuse{pat@proofof@#2}}%
      \addtocounter{counttheorems}{1}
      \expandafter\label{thm_uniq:\pat@uniqlabel}
    }%
  \else
    \AtEndEnvironment{#2}{%
      \xdef\pat@uniqlabel{\thecounttheorems}%
      \xdef\pat@label{\expandafter\expandafter\expandafter
        \@fourthoffour\csname #1\endcsname\space\@currentlabel}%
      \xdef\pat@proofof{\@nameuse{pat@proofof@#2}}%
      \addtocounter{counttheorems}{1}
      \expandafter\label{thm_uniq:\pat@uniqlabel}
    }%
  \fi
  \@namedef{pat@proofof@#2}{#1}%
}

\globtoksblk\prooftoks{1000}
\newcounter{proofcount}

\NewEnviron{proofatend}{%
  The proof can be found in the Appendix \ref{sec:proof_malicious_appendix}.
  \edef\next{%
    \noexpand\begin{proof}[\pat@proofof\space\noexpand\autoref{thm_uniq:\pat@uniqlabel}]%
      \noexpand\hypertarget{proofatend:\pat@uniqlabel}
      \unexpanded\expandafter{\BODY}}%
  \global\toks\numexpr\prooftoks+\value{proofcount}\relax=\expandafter{\next\end{proof}}
  \stepcounter{proofcount}}

\def\printproofs{%
  \count@=\z@
  \loop
    \the\toks\numexpr\prooftoks+\count@\relax
    \ifnum\count@<\value{proofcount}%
    \advance\count@\@ne
  \repeat}
\makeatother

\newtheorem{theorem}{Theorem}[section]

\fixstatement{theorem}

\newtheorem{lemma}[theorem]{Lemma}

\fixstatement{lemma}
\newtheorem{conjecture}[theorem]{Conjecture}

\newtheorem{definition}{Definition}[section]

\newtheorem*{remark*}{Remark}

\numberwithin{algorithm}{section}

\usepackage{makeidx}
\usepackage{color}
\usepackage{graphicx}
\usepackage{hyperref}
\usepackage{caption}
\usepackage{subcaption}
\usepackage[margin=1in]{geometry}
\usepackage{tikz}
\usepackage{tikzscale}
\usepackage{pgfplots}
\usepackage{pgfplotstable}
\usetikzlibrary {positioning}
\usepackage[makeroom]{cancel}
\usepackage[toc,page]{appendix}

\usepackage{qcircuit}
\usepackage[affil-it]{authblk}
\usepackage{float}
\usepackage{enumitem}
\usepackage[textsize=small]{todonotes}
\usepackage[makeroom]{cancel}
\def\EQ#1{\begin{eqnarray}#1\end{eqnarray}} 

\newcommand{\nocontentsline}[3]{}
\newcommand{\tocless}[2]{\bgroup\let\addcontentsline=\nocontentsline#1{#2}\egroup}

\usepackage{breakcites}
\allowdisplaybreaks 

\title{QFactory: classically-instructed remote secret qubits preparation}

\author[1]{Alexandru Cojocaru}
\author[2]{L\'{e}o Colisson}
\author[1,2]{Elham Kashefi}
\author[1]{Petros Wallden}
\affil[1]{School of Informatics, University of Edinburgh,}
\affil[ ]{10 Crichton Street, Edinburgh EH8 9AB, UK}
\affil[2]{D\'{e}partement Informatique et R\'{e}seaux, CNRS, Sorbonne Universit\'{e},}
\affil[ ]{4 Place Jussieu 75252 Paris CEDEX 05, France}

\begin{document}

\date{}

\maketitle

\begin{abstract}
The functionality of classically-instructed remotely prepared random secret qubits was introduced in (Cojocaru et al 2018) as a way to enable classical parties to participate in secure quantum computation and communications protocols. The idea is that a classical party (client) instructs a quantum party (server) to generate a qubit to the server's side that is random, unknown to the server but known to the client. Such task is only possible under computational assumptions. In this contribution we define a simpler (basic) primitive consisting of only BB84 states, and give a protocol that realizes this primitive and that is secure against the strongest possible adversary (an arbitrarily deviating malicious server). The specific functions used, were constructed based on known trapdoor one-way functions, resulting to the security of our basic primitive being reduced to the hardness of the Learning With Errors problem. We then give a number of extensions, building on this basic module: extension to larger set of states (that includes non-Clifford states); proper consideration of the abort case; and verifiablity on the module level. The latter is based on ``\textit{blind self-testing}'', a notion we introduced, proved in a limited setting and conjectured its validity for the most general case. 
\end{abstract}

\newpage

\tableofcontents

\section{Introduction}

In the coming decades, advances in quantum technologies may cause major shifts in the mainstream computing landscape. In the meantime, we can expect to see quantum devices with high variability in terms of architectures and capacities, the so-called noisy, intermediate-scale quantum (NISQ) devices \cite{preskill2018} (such as those being developed by IBM, Rigetti, Google, IonQ) that are currently available to users via classical cloud platforms. In order to be able to proceed to the next milestone for the utility of these devices in a wider industrial base, the issues of privacy and integrity of the data manipulation must be addressed.

 Early proposals for secure and verifiable delegated quantum computing based on simple obfuscation of data  already exist \cite{apbs03, Childs05, ABE08, BFK09, DKL11, TK12, ACJ13, VLTT13,  MVK15, FK17}. However, these original schemes require a reliable long-distance quantum communication network, connecting all the interested parties, which remains a challenging task.

For these reasons, there has recently been extensive research focusing on the practicality aspect of secure and verifiable delegated quantum computation. One direction is to reduce the required communications by exploiting classical fully-homomorphic-encryption schemes \cite{broadbent2015quantum,dulek2016quantum,alagic2017quantum}, or by defining their direct quantum analogues \cite{liang2015quantum,ouyang2015quantum,tan2016quantum,lai2017statistically}. Different encodings, on the client side, could also reduce the quantum communication \cite{mantri2013optimal,GMMR2013}. However, in all these approaches, the client still requires some quantum capabilities. While no-go results indicate restrictions on which of the above properties are jointly achievable for classical clients  \cite{armknecht2014general,yu2014limitations,ACGK2017,newman2017limitations}, recent breakthroughs based on post-quantum secure trapdoor one-way functions, paved the way for developing entirely new approaches towards fully-classical client protocols for emerging quantum servers. he first such procedures were proposed in \cite{urmila_qfhe} allowing a classical client to securely delegate a universal quantum computation to a remote untrusted server, followed by the work of \cite{brakerski_qfhe}, where the construction achieved stronger security. A similar technique was exploited to derive a classical (non-blind) verification scheme for universal computing  \cite{urmila_verif}.

Our own independent work \cite{qfactory_old}), presented in QCrypt '18, took a modular approach, to directly replace the need for any quantum communication channel between client and server with a computationally (but post-quantum) secure generation of secret and random single qubits (rather than directly obfuscating a target quantum functionalities). It was shown then how a classical client could use this module (referred to as QFactory) to achieve secure delegated universal quantum computing, but potentially also, other functionalities such as multi-party quantum computation. 

Following that philosophy, we present in this paper for the first time a universal yet minimal functionality module that is fully secure and verifiable at the module level and could be used as a black box in other client-server applications to replace the need for a reliable long-distance quantum communication network. However, the price one has to pay, would be a reduction from information-theoretic security (achievable using quantum communication) to post-quantum computational security via our modules. The ultimate vision would be to develop a hybrid network of both classical and quantum communication channels, depending on the desired level of security and the technology development of NISQ devices admitting classical or quantum links \cite{Wehner18}. 

\subsection{Our Contributions}

In \cite{qfactory_old} we defined a classical client - quantum server functionality of delegated pseudo-secret random qubit generator (PSRQG) that can replace the need for quantum channel between parties in certain quantum communication protocols, with the only trade-off being that the protocols would become computationally secure (against \emph{quantum} adversaries). In this paper:

\begin{enumerate}

\item We present a new protocol called Malicious 4-states QFactory in \autoref{sec:qfactory2.0} that achieves the functionality of classically instructed remote secret generation of the states $\{\Ket{0}, \Ket{1}, \Ket{+}, \Ket{-}\}$ (known as the BB84 states), given 2 cryptographic functions: 1) a trapdoor one-way function that is quantum-safe, two-regular and collision resistant and 2) a homomorphic, hardcore predicate. The novelty of this new protocol reflects in both simplicity of construction, as well as enhanced security, namely the protocol is secure against any arbitrarily deviating adversary. The target output qubit set is one of the four BB84 states, states that form the core requirement of any quantum communication protocol. \\
Then, in \autoref{subsec:qfactory2.0_secure}, we present the security of the Malicious 4-states QFactory against any fully malicious server, by proving that the basis of the generated qubits are completely hidden from any adversary, using the properties of the two functions, the security being based on the hardness of the Learning with Errors problem.

\item While the above-mentioned results do not depend on the specific function used, the existence of such functions (with all desired properties) makes the functionality a practical primitive that can be employed as described in this paper. In \autoref{sec:func_cons}, we describe how to construct the two-regular, collision resistant, trapdoor one-way family of functions and the homomorphic, hardcore predicate. 
 Furthermore, we prove using reductions in \autoref{subsec:cons_f} that the resulting functions maintain all the required properties.
 
\item In order to demonstrate the modular construction of the basic Malicious 4-states QFactory, we also present in \autoref{sec:8_states}, a secure and efficient extension to the functionality of generating 8 states, called the Malicious 8-states QFactory protocol (where the security refers to the fact that the basis of the new state is completely hidden). The set of output states $\left\{ \Ket{+_{\theta}} \, | \, \theta \in \{0,\frac{\pi}{4}, ..., \frac{7\pi}{4} \} \right\}$ (no longer within the Clifford group) are used in various protocols, including protocols for verifiable blind quantum computation.

\item While the protocol introduced in \autoref{sec:qfactory2.0} requires (for the security proof) a family of functions having 2 preimages with probability super-polynomially close to 1, we also define in \autoref{sec:qfactory_abort} a protocol named Malicious-Abort 4-states QFactory, that is secure when the functions have 2 preimages with only a constant (greater than $1/2$) probability. Indeed, even if the parameters used for the first category of functions are implicitly used in some protocols \cite{urmila_qfhe}, the second category of functions is strictly more secure and more standard in the cryptographic literature \cite{brakerski_qfhe}. The Malicious-Abort 4-states QFactory protocol is proven secure also for this second category of functions, assuming that the classical Yao's XOR lemma also applies for one-round protocols (with classical messages) with quantum adversaries.
  
\item The Malicious 8-states QFactory can be further extended in order to offer a notion of verification for QFactory in \autoref{sec:verif}, the new protocol being called Verifiable QFactory. We demonstrate that this notion of verifiability of QFactory is suitable, by showing that it is sufficient to obtain \textit{verifiable blind quantum computation}. Such protocol would be the first classical client, verifiable, \emph{blind} quantum computation protocol. \\
We introduce in \autoref{sec:blind_self_testing} a novel framework called \textit{blind self-testing}, which differs from the standard self-testing by replacing the non-locality assumptions for such tests with blindness conditions. 
We describe how this technique can be used to prove the verifiability of QFactory.
Note however, that the security of the Verifiable QFactory \autoref{protocol:vQFactory} is conjectured, while we expect that the full proof would follow using the most general case of the novel notion of \emph{blind self-testing} that we introduced. Finally,
we prove how a (much simpler) i.i.d. blind self-testing is achievable.

\end{enumerate}

\subsection{Overview of the protocols and proofs}

\noindent\textbf{The Protocol.} The general idea is that a classical client communicates with a quantum server instructing him to perform certain actions. By the end of
the interaction, the client obtains a random value
$B=B_1 B_2\in \{00,01,10,11\}$, while the server (if he followed the protocol) ends up with the state $H^{B_1}X^{B_2}\ket{0}$, i.e. with one of the BB84 states. Moreover, the server, irrespective of whether he followed the protocol or how he deviated, cannot guess the value of the (basis) bit $B_1$ any better than making a random guess. 

This module is sufficient to perform (either directly or with simple extensions) multiple secure computation protocols including blind quantum computation. 

To achieve such a task, we require three central elements. Firstly, the quantum operations performed by the server should not be repeatable, in order to avoid letting the (adversarial) server run multiple times these operations and obtain multiple copies of the same output state. That would (obviously) compromise the security since direct tomography of a single qubit is straightforward. This can be achieved if the protocol includes a measurement of many qubits, where the probability of getting twice the same outcome would be exponentially small. The second element is that the server should not be able to efficiently classically simulate the quantum computation that he needs to perform. This is to stop the server from running everything classically and obtaining the explicit classical description of the output state. This is achieved using techniques from post-quantum cryptography and specifically the Learning-With-Errors problem. Lastly, the computation has to be easy to perform for the client, since she needs to know the output state. This asymmetry (easy for client/ hard for server) can be achieved only in the computational setting, where the client has some extra trapdoor information. The protocol requires the following cryptographic primities:
\begin{itemize}
    \item $\mathcal{F}$: a family of 2-regular, collision resistant, trapdoor one-way functions (that can be easily constructed from a family of injective, homomorphic, trapdoor one-way functions $\mathcal{G}$);
    \item $h$: a homomorphic (related to the homomorphic operation of $\mathcal{G}$) and hardcore predicate of the functions $\mathcal{G}$;
\end{itemize}
Given these functions, the protocol steps are given below: The client sends the descriptions of the functions $f_k$ (from the family $\mathcal{F}$) and $h$. The server's actions are described by the circuit given in   \autoref{fig:circuitserver} (see \autoref{sec:qfactory2.0}), classically instructed by the client: prepares one register at $\otimes^n H\ket{0}$ and second register at $\ket{0}^{m}$; then applies $U_{f_k}$ using the first register as control and the second as target; measures the second register in the computational basis, obtains the outcome $y$. Through these steps server produces a superposition of the 2 preimages of $y$ for the function $f_k$. Next, the server is instructed to apply a unitary corresponding to the function $h$ and to measures all but one qubit, which represents the output of the protocol. This last step intuitively magnifies the randomness of all the qubits to this final output qubit. 

Then, it can be proven that, in an honest run, this output state is:
$$ \Ket{out} = H^{B_1}X^{B_2}\ket{0} \text{ , where }$$
$$ B_1 = h(z) \oplus h(z')$$
$$ B_2 = [b_n \oplus \sum_{i = 1}^{n - 1} (z_i \oplus z_i') \cdot b_i \bmod 2] \cdot B_1 \, \oplus \, h(z) \cdot (1 \oplus B_1)$$
and where the 2 preimages of $f_k$ are written as: $x = (z, 0)$ and $x' = (z', 1)$. \\
Therefore, the client can efficiently obtain the description of the output state, namely $B_1$ and $B_2$ by inverting $y$, to obtain the 2 preimages $x$ and $x'$ using his secret trapdoor information $t_k$.

\noindent\textbf{Security.} Informally speaking the desired security property of the module is to prove that the server cannot guess better than randomly the basis bit $B_1$ of what the client has, no matter how the server deviates or what answers he returns. In other words, we prove that given that the client chooses $k$ randomly, then no matter which messages $y$ and $b$ the server returns, he cannot determine $B_1$.
More specifically, using the properties of the two cryptographic functions, we show that the basis of the output state is independent of the messages sent by the server and essentially, the basis is fixed by the client at the beginning of the protocol.

Here it is important to emphasize that the simplicity of our modular construction allow us to make a direct reduction from the above security property to the cryptographic assumptions of our primitives functions $\mathcal{F}$ and $h$. Indeed, from the expression above, we can see that at the end of the interaction the client has recorded as the basis bit the expression $B_1 = h(z) \oplus h(z')$. \\
However, from the properties of the functions, we observe the following:
\begin{itemize}
    \item The 2 preimages of $f_k$ satisfy: $z' = z - z_0$
    \item From the homomorphic property of $h$, we get: $B_1 = h(z_0)$
    \item The basis bit $B_1$ depends on the choice of $z_0$ (randomly chosen by the client) and nothing else. Therefore $B_1$ is independent from the actions of the server.
    \item Finally, $h$ being a hardcore predicate, implies that the server should be unable to guess this irrespective of what answers $y,b$ he returns.
\end{itemize}

\noindent\textbf{The Primitive Construction.} In order to use this module in practise, it is crucial to have functions that satisfy our cryptographic requirements, and explore the choices of parameters that ensure that all these properties are jointly satisfied. Building on the function construction of \cite{qfactory_old} we gave specific choices that achieve these properties. The starting point is the injective, trapdoor one-way family of functions from \cite{MP2012}, where the hardness of the function is derived from the Learning With Errors problem.

\noindent\textbf{The Extended Protocol.} In order to use the above protocol for applications such as blind quantum computing \cite{ubqc}, we need to be able to produce states taken from the (extended) set of eight states $\{\Ket{+_{\theta}}, \theta \in \{0, \frac{\pi}{4}, ..., \frac{7\pi}{4}\}\}$. Importantly, we still need to ensure that the bits corresponding to the basis of each qubits produced, remain hidden. Here we prove how given two states produced by the basic protocol described previously, 
which we denote as $\Ket{in_1}$ and $\Ket{in_2}$, we can obtain a single state from the 8-states set, denoted $\Ket{out}$, ensuring that no information about the bits of the basis of $\Ket{out}$ is leaked\footnote{Note that one of the input states is exactly the output of the basic module, while the second comes from a slightly modified version (essentially rotated in the XY-plane of the Bloch sphere).}.

To achieve this, we need to find an operation (see Figure \ref{fig:gadget} in Section \ref{sec:correct_qfac_2_1}), that in the honest case maps the indices of the inputs to those of the output using a map that satisfies certain conditions.
This relation (inputs/output) should be such that learning anything about the basis of the output state implies learning non-negligible information for the basis of (one) input. This directly means, that any computationally bounded adversary that can break the basis blindness of the output, can use this to construct an attack that would also break the basis blindness of at least one of the inputs, i.e. he would break the security guarantees of the basic module that was proven earlier.

\noindent\textbf{Other Properties.} To further demonstrate the utility of our core module, as a building block for other client-server protocols, one might wish to expand further the desired properties of the basic functionality, as well as further enhancing the security model. The simplicity of our construction allows us to extend our security proof directly to the Abstract Cryptography framework of \cite{AC} (work in preparation). Next, to obtain the verifiability of the module (i.e. imposing an honest behaviour on the server) we propose a generalization of the self-testing, where the non-locality condition is replaced by the blindness property and the analysis is done in the computational setting. To further improve the practicality of the black box call of the QFactory we also present the security against abort scenario that could be achieved based on a quantum version of Yao's XOR Lemma. However, these additional properties require stronger basic assumptions that we leave as an open question to be removed or proven correct separately. \\

\section{Preliminaries}

We are considering protocols secure against quantum
adversaries, so we assume that all the properties of our functions hold for a general Quantum Polynomial Time (QPT) adversary, rather than the usual Probabilistic Polynomial Time (PPT) one. We will denote $\mathcal{D}$ the domain of the functions, while $\mathcal{D}(n)$ is the subset of strings of length $n$.

\noindent The following definitions are for PPT adversaries, however in this paper we will generally use quantum-safe versions of those definitions and thus security is guaranteed against QPT adversaries.

\begin{definition}[One-way]\label{def:one_way_function}
A family of functions $\{f_k : \mathcal{D} \rightarrow \mathcal{R}\}_{k \in \mathcal{K}}$ is \textbf{one-way} if:
\begin{itemize}
\item There exists a PPT algorithm that can compute $f_k(x)$ for any index function $k$, outcome of the PPT parameter-generation algorithm \text{Gen} and any input $x \in \mathcal{D}$;
\item Any PPT algorithm $\cA$ can invert $f_k$ with at most negligible probability over the choice of $k$: \\
  $ \underset{\substack{k \leftarrow Gen(1^n) \\  x \leftarrow \mathcal{D} \\ rc \leftarrow \{0, 1\}^{*}}} \Pr [f(\mathcal{A}(k, f_k(x)) = f(x)] \leq \negl$\\
  where $rc$ represents the randomness used by $\mathcal{A}$
\end{itemize}
\end{definition}

\begin{definition}[Collision resistant]\label{def:collision_resistant}
  A family of functions $\{f_k : \mathcal{D} \rightarrow \mathcal{R}\}_{k \in \mathcal{K}}$ is \textbf{collision resistant} if:
  \begin{itemize}
  \item There exists a PPT algorithm that can compute $f_k(x)$ for any index function $k$, outcome of the PPT parameter-generation algorithm \text{Gen} and any input $x \in \mathcal{D}$;
  \item Any PPT algorithm $\cA$ can find two inputs $x \neq x'$ such that $f_k(x) = f_k(x')$ with at most negligible probability over the choice of $k$: \\
  $\underset{
      \substack{k \leftarrow Gen(1^n) \\ rc \leftarrow \{0, 1\}^{*}}}
    \Pr [\cA(k) = (x,x') \text{such that } x \neq x' \text{ and } f_k(x) =
  f_k(x')] \leq \negl$\\
where $rc$ is the randomness of $\cA$ ($rc$ will be omitted from now).
\end{itemize}

\end{definition}

\begin{definition}[k-regular]\label{def:k_regular}
  
A deterministic function $f \colon \mathcal{D} \rightarrow \mathcal{R}$ is \textbf{k-regular} if $ \, \, \forall y \in \Im(f)$, we have ${|f^{-1}(y)| = k}$.
\end{definition}

\begin{definition}[Trapdoor Function]\label{def:trapdoor_function} 
  A family of functions $\{f_k : \mathcal{D} \rightarrow \mathcal{R} \}$
   is a \textbf{trapdoor function} if:
\begin{itemize}
\item There exists a PPT algorithm {\tt Gen} which on input $1^n$ outputs $(k, t_k)$, where $k$ represents the index of the function;
\item $\{f_k : \mathcal{D} \rightarrow \mathcal{R}\}_{k \in \mathcal{K}}$ is a family of one-way functions;
\item There exists a PPT algorithm {\tt Inv}, which on input $t_k$ (which is called the trapdoor information) output by {\tt Gen}($1^n$) and $y = f_k(x)$ can invert $y$ (by returning all preimages of $y$\footnote{While in the standard definition of trapdoor functions it suffices for the inversion algorithm {\tt Inv} to return one of the preimages of any output of the function, in our case we require a two-regular tradpdoor function where the inversion procedure returns both preimages for any function output.})
with non-negligible probability over the choice of $(k, t_k)$ and uniform choice of $x$.

\end{itemize}
\end{definition}

\begin{definition}[Hard-core Predicate]\label{def:hardcore_predicate}
  A function $hc \colon \mathcal{D} \rightarrow \{0, 1\}$ is a \textbf{hard-core predicate} for a function $f$ if:

  \begin{itemize}
\item There exists a QPT algorithm that for any input $x$ can compute $hc(x)$;
\item Any PPT algorithm $\mathcal{A}$ when given $f(x)$, can compute $hc(x)$ with negligible better than $1/2$ probability: \\
$ \underset{\substack{x \leftarrow \mathcal{D}(n) \\ rc \leftarrow \{0, 1\}^{*}}} \Pr [\mathcal{A}(f(x), 1^n) = hc(x)] \leq \frac{1}{2} + \negl$, where $rc$ represents the randomness used by $\mathcal{A}$;
\end{itemize}
\end{definition}

\noindent The Learning with Errors problem (\LWE{}) can be described in the following way:
\begin{definition}[\LWE{} problem (informal)]\label{def:lwe}
Given $s$, an $n$ dimensional vector with elements in $\mathbb{Z}_q$, for some modulus $q$,  the task is to distinguish between a set of polynomially many noisy random linear combinations of the elements of $s$ and a set of polynomially many random numbers from $\mathbb{Z}_q$. 
\end{definition}
Regev \cite{Regev} and Peikert \cite{Peikert} have given quantum and classical reductions from the average case of \LWE{} to problems such as approximating the length of the shortest vector or the shortest independent vectors problem in the worst case, problems which are conjectured to be hard even for quantum computers.
\begin{theorem}[Reduction \LWE{}, from {{\cite[Therem 1.1]{Regev}}}]
  Let $n$, $q$ be integers and $\alpha \in (0,1)$ be such that $\alpha q > 2\sqrt{n}$. If there exists an efficient algorithm that solves $\LWE{}_{q, \bar{\Psi}_\alpha}$, then there exists an efficient quantum algorithm that approximates the decision version of the shortest vector problem \GapSVP{} and the shortest independent vectors problem \SIVP{} to within $\tilde{O}(n/\alpha)$ in the worst case.
\end{theorem}

\begin{definition}[Function Unitary]\label{def:unitary_function}

For any function $f : A \rightarrow B$ that can be described by a polynomially-sized classical circuit, we define the controlled-unitary $U_f$, as acting in the following way:
\EQ{U_f\ket{x}\ket{y} = \ket{x}\ket{y \oplus f(x)} \, \, \, \forall x \in A \, \, \, \forall y \in B,} 
where we name the first register $\ket{x}$ control and the second register $\ket{y}$ target.
Given the classical description of this function $f$, we can always define a QPT algorithm that efficiently implements $U_f$. 

\end{definition}

\subsection{Notations}

We assume basic familiarity with quantum computing notions.\\
For a state $\Ket{+_{\theta}} = \frac{1}{\sqrt{2}}(\Ket{0} + e^{i\theta} \Ket{1})$, where $\theta \in \{0, \frac{\pi}{4}, ..., \frac{7\pi}{4} \}$, we use the notation:
$$\theta = \frac{\pi}{4} L$$
Additionally, as $L$ is a 3-bit string, we write it as $L = L_1L_2L_3$, where $L_1, L_2, L_3$ represent the bits of $L$. \\
As a result when we refer to the basis of the $\Ket{+_{\theta}}$ state, it is equivalent to referring to the last 2 bits of $L$, thus saying that nothing is leaked about the basis of this state, is equivalent to saying nothing is leaked about the bits $L_2$ and $L_3$. \\
For a set of 4 quantum states $\{ \Ket{0}, \Ket{1}, \Ket{+}, \Ket{-} \}$, we denote the index of each state using 2 bits: $B_1, B_2$, with $B_1 = 0$ if and only if the state is $\ket{0}$ or $\ket{1}$, and $B_2 = 0$ if and only if the state is $\ket{0}$ or $\ket{+}$. We will use interchangeably the Dirac notation and the basis/value notation.\\
In the following sections, we will consider polynomially bounded malicious adversaries, usually denoted by $\cA$. The honest clients will be denoted with the $\pi$ letter, and both honest parties and adversaries can output some values, that could eventually be used in other protocols. To denote that two parties $\pi_A$ and $\cA$ interact in a protocol, and that $\pi_A$ outputs $a$ while $\cA$ outputs $b$, we write $(a,b) \leftarrow (\pi_A \| \pi_B)$ (we may forget the left hand side, or replace variables with underscores ``\_'' if it is not relevant). We can also refer to the values of the classical messages send between the two parties using something like $\Pr[a = \accept \mid (\pi_A \| \cA) ]$, and this probability is implicitly over the internal randomness of $\pi_A$ and $\cA$. To specify a two-party protocol, it is enough to specify the two honest parties $(\pi_A,\pi_B)$. Moreover, if the protocol is just made of one round of communication, we can just write $y \leftarrow \cA(x)$ with $x$ the first message sent to $\cA$, and $y$ the messages sent from $\cA$. Finally, a value with a tilde, such as $\tilde{d}$, represents a guess from an adversary.

\section{The Malicious 4-states QFactory Protocol \label{sec:qfactory2.0}}\label{sec:noabortprotocol}

The Malicious 4-states QFactory Protocol uses a family of functions $\mathcal{F}$ and a function $h$. \\
$\mathcal{F}$ is a family $\{f_k : \mathcal{D} \times \{0, 1\} \rightarrow \mathcal{R}\}_k$ of 2-regular, collision resistant, trapdoor one-way functions (against quantum adversaries). This family of functions is constructed using a family $\mathcal{G} = \{g_k : \mathcal{D} \rightarrow \mathcal{R}\}_k$ of injective, trapdoor one-way, homomorphic functions \footnote{We only require $\mathcal{G}$ to be homomorphic with high probability for a single application of the operation $+_{\mathcal{D}}$ and this would result in $\mathcal{F}$ being 2-regular with high probability (as we prove in appendix \ref{app:pseudohom}).}: \\
There exist 2 operations $"+_{\mathcal{D}}"$ acting on $\mathcal{D}$ and $"+_{\mathcal{R}}"$ acting on $\mathcal{R}$ such that:
\EQ{
g_k(z_1 +_{\mathcal{D}} z_2) = g_k(z_1) +_{\mathcal{R}} g_k(z_2) \, \,  \forall k \, \, \forall z_1, z_2 \in \mathcal{D} 
\label{eq:hom_g}
} 

The function $h: \mathcal{D} \rightarrow \{0, 1\}$ is a hardcore predicate with respect to the function $g_k$ and homomorphic:
\EQ{
h(z_1) \oplus h(z_2) = h(z_2 -_{\mathcal{D}} z_1) \, \, \forall z_1, z_2 \in \mathcal{D}, \\
\text{where ''$-_{\mathcal{D}}$'' is the inverse of the operation ''$+_{\mathcal{D}}$''.} \nonumber
\label{eq:hom_h}
}\

Then,the functions $\mathcal{F}$ are constructed as:
\EQ{
f_k : \mathcal{D} \times \{0, 1\} \rightarrow \mathcal{R} \nonumber \\
f_k(z, c) = g_k(z)  +_{\mathcal{R}} c \cdot y_0
}
where $y_0 = g_k(z_0)$ for a fixed $z_0$.

\begin{algorithm}[H]

\caption{Malicious 4-states QFactory Protocol: classical delegation of the BB84 states} \label{protocol:qfactory_real}
\textbf{Requirements:} \\
Public: The functions $\mathcal{F}$ and $h$ described above. For simplicity, we will represent the sets $\mathcal{D}$ and $\mathcal{R}$, using $n - 1$, respectively $m$ bit strings: $\mathcal{D} = \{0, 1\}^{n - 1}$, $\mathcal{R} = \{0, 1\}^{m}$.\\
\textbf{Stage 1: Preimages superposition} \\
-- Client: runs the algorithm $(k,t_k) \leftarrow \text{Gen}_{\mathcal{F}}(1^n)$. The description of $h$ and $k$ are public inputs (known to any party), while $t_k$ is the private input of the Client.
-- Client: instructs Server to prepare one register at $\otimes^n H\ket{0}$ and second register initiated at $\ket{0}^{m}$.

-- Client: sends $k$ to Server and the Server applies $U_{f_k}$ using the first register as control and the second as target.

-- Server: measures the second register in the computational basis, obtains the outcome $y$. Here, in an honest run, the Server would have a state ${(\ket{x}+\ket{x'}) \otimes \ket{y}}$ with $f_k(x)=f_k(x')=y$ and $y\in \Im f_k$. \\
\textbf{Stage 2: Output preparation} \\
-- Server: applies $U_{h}$ using all but the last qubit of the preimage state $\ket{x}+\ket{x'}$ as control and another qubit iniated at $\ket{0}$ as target. Then, measures all the qubits, but the first one (the target) in the $\{\frac{1}{\sqrt{2}}(\Ket{0} \pm \Ket{1})\}$ basis, obtaining the outcome $b = (b_1, ..., b_{n})$. Now, the Server returns both $y$ and $b$ to the Client. \\
-- Client: using the trapdoor $t_k$ computes $x,x'$ and obtains the classical description of the Server's state. \\

\textbf{Output:} If the protocol is run honestly, the state that the Server has produced is:
\EQ{
\Ket{Output} = X^{h(z)} Z^{b_n + \sum_{i = 1}^{n - 1} (z_i \oplus z_i') \cdot b_i} H^{h(z) \oplus h(z')}\Ket{0}
}
where $z$ and $z'$ are the strings obtained from $x$, respectively $x'$ by removing the last bit, while $z_i$ is the $i$-th bit of $z$.

\end{algorithm}

\subsection{Correctness of Malicious 4-states QFactory \label{subsec:qfactory2.0_correct}}

In an honest run, the description of the output state of the protocol depends on the measurement results $y \in Im(f_k)$, and $b$, but also on the 2 preimages $x$ and $x'$ of $y$. Without loss of generality  we assume that these 2 preimages differ in the last bit
We notice, that for our construction of the family $\mathcal{F}$, the 2 unique preimages of the same image differ in the last bit, therefore we will use the notation: $x = (z, 0)$ and $x' = (z', 1)$. \\

The output state of Malicious 4-states QFactory belongs to the set of states $\{\Ket{0}, \Ket{1}, \Ket{+}, \Ket{-}\}$ and its exact description is the following:
\begin{theorem}
In an honest run, the Output state of the Malicious 4-states QFactory Protocol is a BB84 state whose basis is $B_1 = h(z) \xor h(z')$, i.e.: \\
\EQ{
\ket{Output} = H^{B_1}X^{B_2} \ket{0}
\label{eq:correct_output}
}
where:
\EQ{
B_1 &=& h(z) \oplus h(z') \\
B_2 &=& [b_n \oplus \sum_{i = 1}^{n - 1} (z_i \oplus z_i') \cdot b_i \bmod 2] \cdot [h(z) \oplus h(z')] \oplus \{h(z) \cdot [1 \oplus h(z) \oplus h(z')]\} \nonumber \\ 
\label{eq:expressions_exp}
}
\label{thm:correctness}
\end{theorem}

\begin{proof}

The operations performed by the quantum server, can be described as follows:

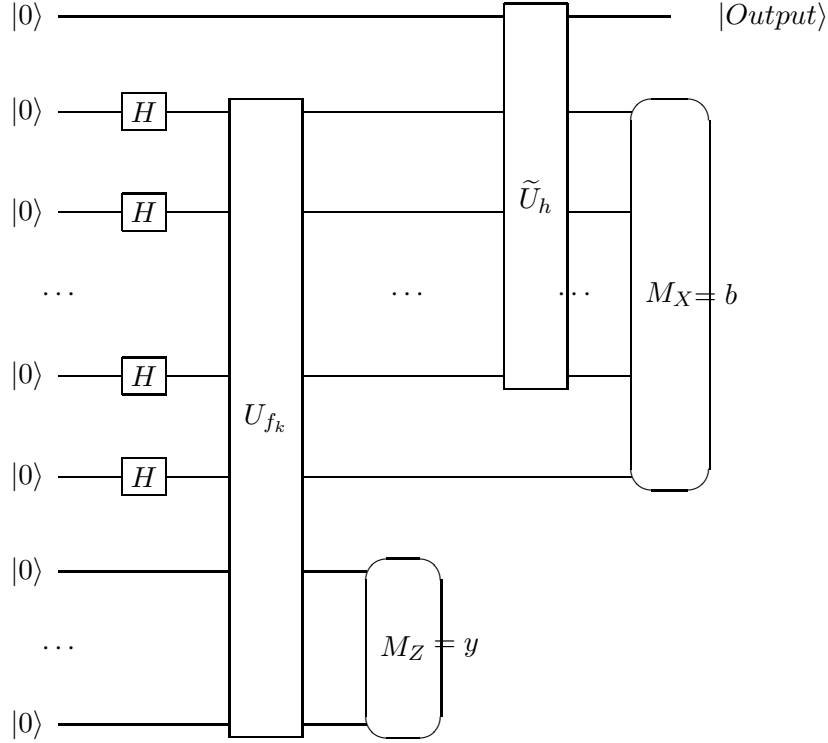
\begin{figure}[ht!]
\[
\Qcircuit{
	& \lstick{\ket{0}} &   \qw &     \qw &                     \qw & \multigate{4}{\widetilde{U}_{h}} & \qw & \Ket{Output} \\
	& \lstick{\ket{0}}  & \gate{H} & \multigate{7}{U_{f_k}} &  \qw & \ghost{\widetilde{U}_{h}} & \multimeasure{4}{M_{X}}    \\
	& \lstick{\ket{0}}  & \gate{H} & \ghost{U_{f_k}} &  \qw & \ghost{\widetilde{U}_{h}} & \ghost{M_{X}}  \\
	&  \cdots  &  \hspace{70mm} \cdots &  \hspace{82mm} \cdots &  \hspace{83mm} = b          \\
	& \lstick{\ket{0}}  & \gate{H} & \ghost{U_{f_k}} &  \qw & \ghost{\widetilde{U}_{h}} & \ghost{M_{X}}  \\
    & \lstick{\ket{0}}  & \gate{H} & \ghost{U_{f_k}} &  \qw & \qw  & \ghost{M_{X}}  \\
    & \lstick{\ket{0}} & \qw &       \ghost{U_{f_k}} & \multimeasure{2}{M_{Z}} \\
    & \cdots & \hspace{83mm} = y             \\
    & \lstick{\ket{0}} & \qw &       \ghost{U_{f_k}} & \ghost{M_{Z}} 
}
\]
\caption{The circuit computed by the Server}
\label{fig:circuitserver}
\end{figure}

\EQ{
\Ket{0} \otimes \Ket{0^n} \otimes \Ket{0^m} &\xrightarrow{I_2 \otimes H^{\otimes n} \otimes {I_2}^{\otimes m} }& \nonumber \\
\Ket{0} \otimes (\sum_{x \in Dom(f_k)} \Ket{x}) \otimes \Ket{0^m} &\xrightarrow{I_2 \otimes U_{f_k} }& \nonumber \\
\Ket{0} \otimes (\sum_{x \in Dom(f_k)} \Ket{x} \otimes \Ket{f_k(x)}) &\xrightarrow{\text{2-regularity of } f_k }& \nonumber \\
\Ket{0} \otimes (\sum_{y \in Im(f_k)} (\Ket{x} + \Ket{x'}) \otimes \Ket{y}) &\xrightarrow{I_2 \otimes {I_2}^{\otimes n} \otimes M_Z^{\otimes m} }& \nonumber \\
\Ket{0} \otimes (\Ket{x} + \Ket{x'}) \otimes \Ket{y} = \Ket{0} \otimes (\Ket{z}\Ket{0} + \Ket{z'}\Ket{1}) \otimes \Ket{y} &\xrightarrow{\widetilde{U}_{h} \otimes {I_2}^{\otimes m} }& \nonumber \\
(\Ket{h(z)} \otimes \Ket{z}\Ket{0} + \Ket{h(z')} \otimes \Ket{z'}\Ket{1}) \otimes \Ket{y} &\xrightarrow{I_2 \otimes M_X^{\otimes n} \otimes {I_2}^{\otimes m} }& \nonumber \\
\Ket{Output} \otimes \Ket{b_1} ... \otimes \Ket{b_n} \otimes \Ket{y} &\Rightarrow& \nonumber \\
\Ket{Output} = X^{h(z)}Z^{b_n + \sum_{i = 1}^{n - 1} (z_i \oplus z_i') \cdot b_i}H^{h(z) \oplus h(z')}\Ket{0}
}
where $\widetilde{U}_{h}$ is defined as $U_{h}$ acting on the first register as target and second register as control: $\Ket{0}\Ket{z}\Ket{c} \xrightarrow{\widetilde{U}_{h}} \Ket{h(z)}\Ket{z} \Ket{c}$.

The sever initially prepares the state $\Ket{0^n} \otimes \Ket{0^m}$, where we will call the first register the preimage register, and the second one the image register. \\ 
After applying $U_{f_k}$ we obtain the state $\sum_{x \in Dom(f_k)} \Ket{x}\Ket{f_k(x)}$. Using the 2-regularity property of $f_k$, after measuring the second register (in the computational basis) and obtaining the measurement result $y \in Im(f_k)$, the state can be expressed as $(\Ket{x} + \Ket{x'}) \otimes \Ket{y}$, where $x$ and $x'$ are the 2 unique preimages of $y$. By omitting the image register and by initializing another qubit in the $\Ket{0}$ state and using the above notation, the input to the unitary $\widetilde{U}_{h}$ can be written as:
\EQ{
(\Ket{z}\otimes \Ket{0} + \Ket{z'} \otimes \Ket{1}) \otimes \Ket{0}
}
$\widetilde{U}_{h}$ is basically $U_{h}$ acting on the first and third register ($\Ket{z}\Ket{c}\Ket{0} \xrightarrow{\widetilde{U}_{h}} \Ket{z}\Ket{c} \Ket{h(z)}$), and after we apply it, we obtain the state:
\EQ{
(\Ket{z} \otimes \Ket{0} \otimes \Ket{h(z)} + \Ket{z'} \otimes \Ket{1} \otimes \Ket{h(z')}
}
As a final step, we measure all but the last qubit of this state in the $\{\frac{1}{\sqrt{2}}(\Ket{0} \pm \Ket{1})\}$ basis (obtaining the measurement result string $b$), and the unmeasured qubit represents the output of the Malicious 4-states QFactory protocol. 

\end{proof}

It can be noticed that, in an honest run of the protocol, using $y$ and the trapdoor information of the function $f_k$, the Client obtains $x$ and $x'$ and thus can efficiently determine what is the output state that the Server has prepared. \\
In the next section, we prove that no malicious adversary can distinguish between the 2 possible bases $\{\Ket{0}$, $\Ket{1}\}$ and $\{\Ket{+}$, $\Ket{-}\}$ of the output qubit, or equivalently distinguish whether $B_1$ is $0$ or $1$.

\subsection{Security against Malicious Adversaries of Malicious 4-states QFactory \label{subsec:qfactory2.0_secure}}

In any run of the protocol, honest or malicious, the state that the client believes that the server has, is given by Thereom \ref{thm:correctness}. \\
Therefore, the task that a malicious server wants to achieve, is to be able to guess, as good as he can, the index of the output state that the client (based on the public communication) thinks the server has produced. In particular, in our case, the server needs to guess the bit $B_1$ (corresponding to the basis) of the (honest) output state.

\begin{definition}[4 states basis blindness]\label{def:4basisblind}
  We say that a protocol $(\pi_A, \pi_B)$ achieves \textbf{basis-blindness} with respect to an ideal set of 4 states $S$ if:
  \begin{itemize}
  \item $S$ is the set of states that the protocol outputs, i.e.:
    \[\Pr[\ket{\phi} \in S \mid (\_,\ket{\phi}) \leftarrow (\pi_A \| \pi_B) ]=1\]
  \item and if no information is leaked about the index bit $B_1$ of the output state of the protocol, i.e:
  \[\forall \cA, \Pr[ B_1 = \tilde{B_1} \mid (B_1B_2, \tilde{B_1}) \leftarrow (\pi_A \| \cA)] \leq 1/2 + \negl \]
  \end{itemize}
\end{definition}

\begin{theorem} \label{thm:security_qfac_2_0}
Malicious 4-states QFactory satisfies $4$-state basis blindness with respect to the ideal set of states $\{\Ket{0}, \Ket{1}, \Ket{+}, \Ket{-}\}$.
\end{theorem}
\begin{proof}
From Theorem \ref{thm:correctness}, we notice that the basis of the output qubit is $\{\Ket{0}, \Ket{1}\}$ if $B_1 = h(z) \oplus h(z') = 0$ and the basis is $\{\Ket{+}, \Ket{-}\}$ if $B_1 = h(z) \oplus h(z') = 1$. \\
Therefore, to hide the basis of the output qubit (against any malicious adversary) is equivalent to hiding the value of the bit $B_1 = h(z) \oplus h(z')$. \\

Next, we will prove that $h(z) \oplus h(z')$ is indistinguishable from a random bit, using the properties of the functions $\mathcal{F}$ and $h$. \\
As a first step, we recall that our function $f_k$ is constructed using the injective, homomorphic trapdoor one-way function $g_k$:
\EQ{
f_k : \mathcal{D} \times \{0, 1\} \rightarrow \mathcal{R} \nonumber \\
f_k(z, c) = g_k(z)  +_{\mathcal{R}} c \cdot y_0
}
where $y_0 = g_k(z_0)$ for a fixed $z_0$ \footnote{At the beginning of each run of Malicious 4-states QFactory, the Client chooses a new $z_0$ uniformly at random from the domain of $g_k$ and then he sends $y_0 = g_k(z_0)$ to the Server, in order for the Server to compute $U_{f_k}$.}. The proof of this construction can be found in Thm 6.1, in \cite{qfactory_old}). \\

Using the homomorphic property of $g_k$ from \autoref{eq:hom_g}, we can rewrite $f_k$ as:
\EQ{
f_k(z, c) = g_k(z) +_{\mathcal{R}} c \cdot y_0  = g_k(z) +_{\mathcal{R}} c \cdot g_k(z_0) = g_k(z +_{\mathcal{D}}  (c \cdot z_0))
}
And as $g_k$ is injective, for any image $y$ of the function $f_k$, the 2 preimages of $y$ will be of the form $x = (z, 0)$ and $x' = (z +_{\mathcal{D}} z_0, 1)$ and equivalently $z' = z +_{\mathcal{D}} z_0$.\footnote{As the functions $\mathcal{G}$ are homomorphic with high probability, we derive 2-regularity and thus the existence of 2 preimages of this form with high probability as we prove in \autoref{lemma:approx_two_reg}.}\\
Then, using the homomorphic property of $h$, we deduce that the basis bit $B_1$ of the output state is equal to:
\EQ{
B_1 =  h(z) \oplus h(z') =  h(z' -_{\mathcal{D}} z) = h(z_0)
}

Finally, to prove that $B_1$ is completely hidden we will use the hardcore property of $h$. \\
To do that, we proceed with a proof by contradiction. We assume there exists a QPT adversary $\mathcal{A}$ that can determine $B_1$ with a non-negligible advantage and construct a QPT adversary $\mathcal{A}'$ that breaks the hardcore predicate property of $h$ with the same non-negligible advantage, which completes the proof.

\procedure [linenumbering]{ $Hardcore_{\mathcal{A}', h } $}{
(k, t_k) \sample {Gen}_{\mathcal{G}}(1^n) \\
z_0 \sample \mathcal{D} \\
k' \gets (k, g_k(z_0)) \\
t_k' \gets (t_k, z_0) \\
B_1' \gets \mathcal{A}(k', h) \\
\pccomment{ $B_1$ is the index of the outcome} \\ \pccomment{ that we assumed $\mathcal{A}$ is able to determine} \\
\pcreturn B_1'
\pccomment{as $B_1 = h(z_0)$, this is equivalent to breaking the hardcore property} \\
\pccomment{ with the same non-negligible advantage of $\mathcal{A}$ to determine $B_1$} \\
}

Remark: in the run of the Malicious 4-states QFactory protocol, the adversary/server has no access to the abort/accept bit, specifying whether the Client wants to abort the protocol after receiving the image $y$ from the server (the abort occurs when $y$ does not have exactly two preimages). See \autoref{sec:qfactory_abort} to see how we address this issue.

\end{proof}

 \section{Function Implementation \label{sec:func_cons}}

To complete the construction of Malicious 4-states QFactory, we must find functions $\mathcal{F}$ and $h$ satisfying the properties described in \autoref{sec:qfactory2.0}. \\

The starting point is the injective, trapdoor one-way family of functions $\bar{\mathcal{G}} = \{\bar{g}_K : \mathbb{Z}_q^n \times {\chi}^m \rightarrow \mathbb{Z}_q^m\}_K $ from \cite{MP2012} (where $\chi$ defines the set of integers bounded in absolute value by some value $\mu$ which will be defined later):
\EQ{
\bar{g}_K(s,e) = Ks + e \bmod q
}
From this we construct the family of functions $\bar{\mathcal{F}} = \{ \bar{f}_K : \mathbb{Z}_q^n \times {\chi}^m \times \{0, 1\} \rightarrow \mathbb{Z}_q^m \}$:  
\EQ{
\bar{f}_K(s,e, c) = Ks + e + c \cdot \bar{y}_0 \bmod q
}
where $\bar{y}_0 = \bar{g}_K(s_0,e_0)$, with $s_0$ chosen at random from $\mathbb{Z}_q^n$ and $e_0$ chosen at random from ${\chi'}^m$ ($\chi'$ is the set of integers bounded in absolute value by some value $\mu' < \mu$) .\\

\begin{lemma}[from \cite{qfactory_old}]
The family of functions of $\bar{\mathcal{F}}$ is 2 regular with high (constant) probability, trapdoor, one-way and collision resistant (all these properties are true even against a quantum attacker), assuming that there is no quantum algorithm that can efficiently solve $SIVP{}_\gamma$ for $\gamma = \poly[n]$, for the following choices of parameters:
\EQ{
q &=& 2^{5\ceil{log(n)} + 21} \nonumber \\
m &=& 23n + 5n\ceil{log(n)} \nonumber \\
\mu &=& 2mn\sqrt{23+5log(n)} \nonumber \\
\mu' &=& \mu / m
\label{eq:params}
}
\label{lemma:params}
\end{lemma}

\subsection{Construction of $\mathcal{F}$ \label{subsec:cons_f}}

Now, we are able to define the 2 families of functions $\mathcal{G}$, $\mathcal{F}$:
$$g_K : \mathbb{Z}_q^n \times {\chi}^m \times \{0, 1\} \rightarrow \mathbb{Z}_q^m$$
\EQ{
g_K(s,e,d) = \bar{g}_K(s,e) + d \cdot \begin{pmatrix}
  {\frac{q}{2}}\\
  0\\
  \vdots\\
  0\\
\end{pmatrix} \bmod q = 
Ks + e + d \cdot \begin{pmatrix}
  {\frac{q}{2}}\\
  0\\
  \vdots\\
  0\\
\end{pmatrix} \bmod q 
\, \, \, \, \, \, \, \,
\footnote{in this section we will use the notation $v = \begin{pmatrix}
  {\frac{q}{2}}\\
  0\\
  \vdots\\
  0\\
\end{pmatrix}$.}
}
\EQ{
f_K : \mathbb{Z}_q^n \times {\chi}^m \times \{0, 1\} \times \{0, 1\} \rightarrow \mathbb{Z}_q^m \nonumber \\
f_K(s, e, c, d) = Ks + e + c \cdot y_0 + d \cdot v \bmod q
}
where $y_0 = g_K(s_0,e_0,d_0)$, with $s_0$ chosen at random from $\mathbb{Z}_q^n$, $e_0$ chosen at random from ${\chi'}^m$ and $d_0$ chosen at random from $\{0, 1\}$ .  \\ \\

Before defining the function $h$, we must ensure that $g_K$ has the same properties as $\bar{g}_K$ (homomorphic, injective and one-way), which would also imply that $f_K$ has the same properties as $\bar{f}_K$ (quantum-safe, two-regular, trapdoor one-way).

\begin{lemma}
If $\bar{\mathcal{G}}$ is a family of homomorphic functions, then $\mathcal{G}$ is also a family
of homomorphic functions.
\label{lemma:hom}
\end{lemma}

\begin{proof}
To prove that $g_K$ is homomorphic, we notice that:
\EQ{
g_K(s_1,e_1,d_1) + g_K(s_2,e_2,d_2)
= \bar{g}_K(s_1,e_1) + d_1 \cdot v + 
\bar{g}_K(s_2,e_2) + d_2 \cdot v \bmod q \nonumber \\
 = \bar{g}_K(s_1 + s_2 \bmod q, e_1 + e_2) + (d_1 + d_2) \cdot v \bmod q \nonumber \\
 = g_K(s_1 + s_2 \bmod q, e_1 + e_2, d_1 \oplus d_2) \nonumber
}
where for the last equality we used the fact that if $d_1, d_2 \in \{0, 1\}$, then  $d_1 \cdot \frac{q}{2} + d_2 \cdot \frac{q}{2} \bmod q = (d_1 \oplus d_2) \cdot \frac{q}{2} \bmod q$. \\

\end{proof}

We make the following remark: the proof of \autoref{lemma:hom} is constructed for the case when $\bar{g}$ is perfectly homomorphic, but it also holds in the case when $\bar{g}$ is homomorphic with high probability, resulting in $g$ being homomorphic with the same high probability.

\begin{lemma}
If $\bar{\mathcal{G}}$ is a family of one-way functions, then $\mathcal{G}$ is also a family
of one-way functions.
\label{lemma:oneway}
\end{lemma}

\begin{proof}
To prove the \textbf{one-wayness} of $g$, we are going to reduce it to the one-wayness of $\bar{g}$. \\

Thus, we assume there exists a QPT adversary $\mathcal{A}$ that can invert $g$ with probability $P$ and we construct a QPT adversary $\mathcal{A'}$ inverting $\bar{g}$ with the same probability $P$.

\procedure [linenumbering]{ $Invert_{\mathcal{A}', K}(y) $} {
d \sample \{0,1\} \\
y' \gets y + d \cdot v \\ 
(s', e', d') \gets \mathcal{A}_K(y')\\
\pcreturn (s', e') 
}

\end{proof}

\begin{lemma}
If $\bar{\mathcal{G}}$ is a family of injective functions, then $\mathcal{G}$ is a family
of injective functions.
\label{lemma:injective}
\end{lemma}

\begin{proof}
To prove this, we will use the injectivity property of the function $\bar{g}$.
\EQ{
g_K(s, e, d) = \bar{g}_K(s, e) + d \cdot v \nonumber
}
Assume there exist 2 tuples $(s_1, e_1, d_1)$ and $(s_2, e_2, d_2)$ such that $g(s_1, e_1, d_1) = g(s_2, e_2, d_2)$. To prove that $g$ is injective we must show that $(s_1, e_1, d_1) = (s_2, e_2, d_2)$. \\
 This is equivalent to:
\EQ{
 \bar{g}_K(s_1, e_1) + d_1 \cdot v = \bar{g}_K(s_2, e_2) + d_2 \cdot v \nonumber \\
 \bar{g}_K(s_1, e_1) - \bar{g}_K(s_2, e_2) + (d_1 - d_2) \cdot v = 0 
}
Now, if $d_1 = d_2$, then we have that $\bar{g}(s_1, e_1) - \bar{g}(s_2, e_2) = 0$, and because $\bar{g}$ is injective, this would imply that $s_1 = s_2$ and $e_1 = e_2$ and thus $g$ would also be injective. \\
Let us suppose that $d_1 \neq d_2$ and we will prove that this is impossible. Without loss of generality we can assume that $d_1 = 0$ and $d_2 = 1$. \\

Thus, we want to show that it is impossible to have $(s_1, e_1)$ and $(s_2, e_2)$ such that:
\EQ{
\bar{g}_K(s_1, e_1) - \bar{g}_K(s_2, e_2) = v
}
Equation is equivalent to:
\EQ{
Ks_1 + e_1 - Ks_2 - e_2 = \begin{pmatrix}
  {\frac{q}{2}}\\
  0\\
  \vdots\\
  0\\
\end{pmatrix}
\bmod q
\label{eq:inj}
} 

This can be rewritten as:
\EQ{
K_{1, 1}(s_{1,1} - s_{2,1})+ ... + K_{1, n}(s_{1, n} - s_{2, n}) + (e_{1, 1} - e_{2, 1}) = \frac{q}{2} \bmod q \label{eq:first_line_sys}
}
\EQ{
K_{i, 1}(s_{1,1} - s_{2,1})+ ... + K_{i, n}(s_{1, n} - s_{2, n}) + (e_{1, i} - e_{2, i}) = 0 \bmod q  \forall i = 2,..,m \label{eq:rest_sys}
}
where $s_{1, i}$ and $s_{2, i}$ are the $i$-th bits of $s_1$ and $s_2$ respectively and $e_{1, i}$ and $e_{2, i}$ are the $i$-th bit of $e_1$, respectively $e_2$. \\

Now, as the function $\bar{g}_K: \mathbb{Z}_q^n \times {\chi}^m \times \{0, 1\} \rightarrow \mathbb{Z}_q^m $, where $K \leftarrow \mathbb{Z}_q^{m \times n}$ is injective, the following function is also injective: \\
$\bar{g_1}_{K_1}: \mathbb{Z}_q^n \times {\chi}^{m-1} \times \{0, 1\} \rightarrow \mathbb{Z}_q^{m-1} $, where $K_1 \leftarrow \mathbb{Z}_q^{m - 1 \times n}$ and where $\bar{g}$ and $\bar{g_1}$ have the exact same definition: \\
$$\bar{g_1}_{K_1}(s, e) = K_1s + e \bmod q$$
More specifically, the only difference between the 2 functions is the change of dimension from $m$ to $m - 1$, but as the injectivity proof from \cite{MP2012} holds for any $m = \Omega(n)$, then $\bar{g_1}$ is also injective. 

Now, consider the matrix $K_1$ obtained from $K$ by removing the first line. As shown above,  $\bar{g_1}_{K_1}$ is an injective function, thus from \autoref{eq:rest_sys}, we get that:
\EQ{
s_1 &=& s_2 \\
e_{1, i} &=& e_{2, i} \, \, \forall \,  i = 2,..,m
}
Now as $s_1 = s_2$, from \autoref{eq:first_line_sys}, we obtain $e_{1, 1} - e_{2, 1} = \frac{q}{2}$.\\ However, from the domain of $\bar{g}$, we have that: $|e_{1, 1} - e_{2, 1}| < 2\mu < \frac{q}{2}$ (where for the last inequality we used Lemma \ref{lemma:params}). Contradiction.

\end{proof}

\subsection{Construction of function $h$ \label{subsec:cons_h}}

At this stage, we can define $h$:
\EQ{
h: \mathbb{Z}_q^n \times \Z_q^m &\times& \{0, 1\} \rightarrow \{0, 1\} \nonumber \\
h(s, e, d) &=& d
}
\begin{theorem}
The function $h$ is homomorphic and hardcore predicate (with respect to the one-way function $g$). 
\label{thm:h_hom_hp}
\end{theorem}

\begin{proof}

We can easily notice the homomorphicity property of the defined function $h$:
\EQ{
h(s_1, e_1, d_1) \oplus h(s_2, e_2, d_2) = d_1 \oplus d_2 \nonumber \\
= h(s_1 + s_2 \bmod q, e_1 + e_2 \bmod q, d_1 \oplus d_2)
}
To prove that $h$ is a hardcore predicate of $g_K$, we must prove that for any QPT adversary $\mathcal{A}$, we have:
\EQ{
 &\Pr_{\substack{s \sample \mathbb{Z}_q^n \\ e \sample \chi^m \\ d \sample \{0, 1\}}}  [\mathcal{A}(1^{\lambda}, K, g_K(s, e, d)) = h(s, e, d)] \leq \frac{1}{2} + negl(\lambda) 
}
Using the definitions of the 2 functions, we can express it as:
\EQ{
 &\Pr_{\substack{s \sample \mathbb{Z}_q^n \\ e \sample \chi^m \\ d \sample \{0, 1\}}}  [\mathcal{A}(1^{\lambda}, K, Ks + e + d \cdot v)  = d] \leq \frac{1}{2} + negl(\lambda) 
}
This is equivalent to prove that the distributions $D_1 = \{K, Ks + e\}$ and $D_2 = \{K, Ks + e + v\}$ are indistinguishable\footnote{It is also easy to write an explicit reduction}, or equivalently:
\begin{equation}\{ K_i, \langle K_i, s \rangle + e_i \}_{i = 1}^{m} \overset{c}{\approx} \{ K_i, \langle K_i, s \rangle + e_i + v_i \}_{i = 1}^{m}  
\label{eq:indisting}
\end{equation}
Using the decisional LWE assumption we already know that when $u_i$ are uniform chosen from $\mathbb{Z}_q$, we have \footnote{this holds because the parameters (fully given in \cite[Lemma 6.9]{qfactory_old}) of the function are chosen to make $y_0$ indistinguishable from a random vector by a direct reduction to LWE}:
\EQ{
\{ K_i, \langle K_i, s \rangle + e_i \}_{i = 1}^{m} \overset{c}{\approx}  \{ K_i, u_i\}_{i = 1}^{m}
}
Then, as $v$ is a fixed vector, we also have that:
\EQ{
\{ K_i, \langle K_i, s \rangle + e_i + v_i \}_{i = 1}^{m} \overset{c}{\approx}  \{ K_i, u_i\}_{i = 1}^{m} \overset{c}{\approx} \{ K_i, \langle K_i, s \rangle + e_i \}_{i = 1}^{m}
}
which completes the proof.

\end{proof}

\section{The Malicious 8-states QFactory Protocol \label{sec:8_states}}

In order to use the Malicious 4-states QFactory Protocol functionality for applications such as blind quantum computing \cite{ubqc}, we need to be able to produce states taken from the set $\{\Ket{+_{\theta}}, \theta \in \{0, \frac{\pi}{4}, ..., \frac{7\pi}{4}\}\}$, always ensuring that the bases of these qubits remain hidden. Here we prove how by obtaining two states of Malicious 4-states QFactory Protocol, we can obtain a single state from the 8-states set, while no information about the bases of the new output state is leaked.

To achieve this, we need to find an operation, that in the honest case maps the correct inputs to the outputs, in such a way, that the index of the output state corresponding to the basis, is directly related with the bases bits of the input states. This relation should be such that learning anything about the basis of the output state implies learning non-negligible information for the input. This directly means, that any computationally bounded adversary that breaks the 8-states basis blindness of the output, also breaks the 4-states basis blindness of at least one of the inputs.\footnote{Here it is worth pointing out that a similar result (in a more complicated method) was achieved in \cite{vedran_elham}. That technique however, is applied in the information theoretic setting.}

\begin{algorithm}[H]
\caption{Malicious 8-states QFactory} \label{protocol:QFactory4to8}
\textbf{Requirements:} Same as in Protocol \autoref{protocol:qfactory_real}\\
\textbf{Input:} Client runs 2 times the algorithm ${Gen}_{\mathcal{F}}(1^n)$, obtaining $(k^1, t^1_k), (k^2, t^2_k)$. Client keeps $t_k^1, t_k^2$ private.\\ 
\textbf{Protocol:}\\
-- Client: runs the Malicious 4-states QFactory algorithm to obtain a state $\ket{\quin_1}$ and a "rotated" Malicious 4-states QFactory to obtain a state $\ket{\quin_2}$ (by a rotated Malicious 4-states QFactory we mean a Malicious 4-states QFactory, but where the last set of measurements in the $\Ket{\pm}$ basis (\autoref{fig:circuitserver}) is replaced by a set of measurements in the $\Ket{\pm_{ \frac{\pi}{2}} }$ basis). \\ 
-- Client: records measurement outcomes $(y^1,b^1)$, $(y^2, b^2)$ and computes and stores the corresponding  indices of the output states of the 2 Malicious 4-states QFactory runs: $(B_1, B_2)$ for $\ket{\quin_1}$ and $(B_1', B_2')$ for $\ket{\quin_2}$  \\
-- Client: instructs the server to apply the Merge Gadget (\autoref{fig:gadget}) on the states $\ket{\quin_1}$, $\ket{\quin_2}$. \\
-- Server: returns the 2 measurement results $s_1$, $s_2$. \\
-- Client: using $(B_1, B_2)$, $(B_1', B_2')$, $s_1$, $s_2$ computes the index $L = L_1L_2L_3 \in \{0, 1\}^3$ of the output state. \\
\textbf{Output:} If the protocol is run honestly, the state that the Server has produced is:
\EQ{
\ket{\quout} =  X^{{(s_2 + B_2) \cdot B_1}} Z^{B_2' + B_2(1 - B_1) + B_1[s_1 + (s_2 + B_2)B_1']} R(\frac{\pi}{2})^{B_1} R(\frac{\pi}{4})^{B_1'} \Ket{+}
}

\end{algorithm}

\subsection{Correctness of Malicious 8-states QFactory\label{sec:correct_qfac_2_1}}

We prove the existence of a mapping $\mathcal{M}$ (which we will call Merge Gadget),  from 2 states $\ket{\quin_1}$ and $\ket{\quin_2}$, where $\ket{\quin_1} \in \{\Ket{0}, \Ket{1}, \Ket{+}, \Ket{-} \}$ and $\ket{\quin_2} \in \{\Ket{+}, \Ket{-}, \Ket{+_y}, \Ket{-_y} \}$ to a state $\ket{\quout}  = \Ket{+_{L \cdot \frac{\pi}{4}}}$, where $L = L_1L_2L_3 \in \{0, 1\}^3$. \\
Namely, as defined in Protocol \ref{protocol:QFactory4to8}, $\mathcal{M}$ is acting in the following way:
\EQ{
\mathcal{M}(\ket{\quin_1}, \ket{\quin_2}) = M_{X, 2} M_{X, 1}{\wedge Z}_{2,3} {\wedge Z}_{1,2} \left[ \Ket{+_{\frac{\pi}{4}}} \otimes \ket{\quin_1} \otimes \ket{\quin_2} \right]
}

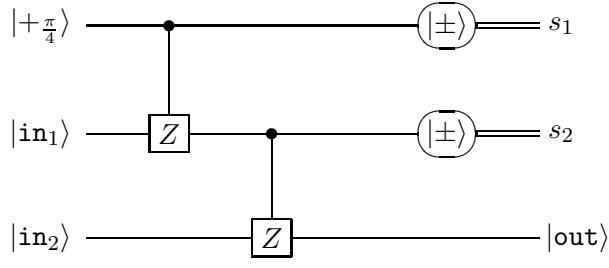
\begin{figure}[ht!]
\[
\Qcircuit{
	& \lstick{\ket{+_{\frac{\pi}{4}}}} &  \ctrl{1}  & \qw  & \qw  & \measure{\ket{\pm}} & \cw \, \, \,  \, \, \,  \, \, \,  s_1
    \\
    & \lstick{\ket{\quin_1}} &     \gate{Z}  & \ctrl{1} &  \qw  & \measure{\ket{\pm}} & \cw \, \, \,  \, \, \,  \, \, \,  s_2
 \\
    & \lstick{\ket{\quin_2}} &                \qw      & \gate{Z}  & \qw & \qw & \qw \, \, \,  \, \, \,  \, \, \, \, \, \, \, \, \, \,  \ket{\quout}
}
\]

\caption{Merge Gadget}
\label{fig:gadget}
\end{figure}

\begin{theorem}
In an honest run, the Output state of the Malicious 8-states QFactory Protocol is of the form $\Ket{+_{L \cdot \frac{\pi}{4}}}$, where $L = L_1L_2L_3 \in \{0, 1\}^3$.

\label{thm:correct_qfac_2_1}
\end{theorem}

\begin{proof}
In an honest run, using the Merge Gadget (\autoref{fig:gadget}) we get:
\EQ{
\mathcal{M}(\ket{\quin_1}, \ket{\quin_2}) = M_{X, 2} M_{X, 1}{\wedge Z}_{2,3} {\wedge Z}_{1,2} \left[ \Ket{+_{\frac{\pi}{4}}} \otimes \ket{\quin_1} \otimes \ket{\quin_2} \right]
}

Using the correctness of Malicious 4-states QFactory (Thereom \ref{thm:correctness}), we have that:
\EQ{
\ket{\quin_1} = H^{B_1}X^{B_2}\Ket{0} \\
\ket{\quin_2} = R(\frac{\pi}{2})^{B_1'}Z^{B_2'}\Ket{+}
}
Thus:
\EQ{
\ket{\quout} = M_{X, 2} M_{X, 1}{\wedge Z}_{2,3} {\wedge Z}_{1,2} \left[ \Ket{+_{\frac{\pi}{4}}} \otimes H^{B_1}Z^{B_2}\Ket{+} \otimes R(\frac{\pi}{2})^{B_1'}Z^{B_2'}\Ket{+} \right]
}
Which is then equivalent to:

\EQ{
\ket{\quout} = R[\pi(B_2' + B_2 + B_1 \cdot (s_1 + s_2)) + \frac{\pi}{2}(B_1' + (B_2 + s_2) \cdot B_1) + \frac{\pi}{4} B_1 ] \Ket{+}
}

As a result, we obtain:

\EQ{
L_1 &=& B_2' \oplus B_2 \oplus [B_1 \cdot (s_1 \oplus s_2)] \nonumber \\
L_2 &=&  B_1' \oplus [(B_2 \oplus s_2) \cdot B_1] \label{eq:expressions_l_2}\\
L_3 &=&  B_1
\label{eq:expressions_l_3}
}

\end{proof}

It can also be noticed that, in an honest run of Malicious 8-states QFactory, the client can efficiently determine $L$: using $b^1, b^2, y^{1}, y^{2}$ and the trapdoors $t_k^{1}, t_k^{2}$, he first obtains $(B_1, B_2)$ and $(B_1', B_2')$, and after receiving $s_1, s_2$, he determines the description of the state prepared by the server.

\subsection{Security against Malicious Adversaries of Malicious 8-states QFactory \label{sec:security_qfac_2_1}}

In any run of the protocol, honest or malicious, the state that the client believes that the server has, is given by \autoref{thm:correct_qfac_2_1}. \\
Therefore, as in the case of Malicious 4-states QFactory, the task that a malicious server wants to achieve, is to be able to guess, as good as he can, the index of the output state that the client thinks the server has produced. In particular, in our case, the server needs to guess the bits $L_2$ and $L_3$ (corresponding to the basis) of the (honest) output state.

\begin{definition}[8 states basis blindness]\label{def:8basisblind}
  Similarly, we say that a protocol $(\pi_A, \pi_B)$ achieves \textbf{basis-blindness} with respect to an ideal set of 8 states $S$ if:
  \begin{itemize}
  \item $S$ is the set of states that the protocol outputs, i.e.:
    \[\Pr[\ket{\phi} \in S \mid (\_,\ket{\phi}) \leftarrow (\pi_A \| \pi_B) ]=1\]
  \item and if no information is leaked about the index bit $B_1B_2$ of the output state of the protocol, i.e:
  \[\forall \cA, \Pr[ B_1 = \tilde{B_1} \text{ and } B_2 = \tilde{B_2} \mid (B_1B_2B_3, \tilde{B_1}\tilde{B_2}) \leftarrow (\pi_A \| \cA)] \leq 1/4 + \negl \]
  \end{itemize}
\end{definition}

\begin{theorem}\label{thm:security_qfac_2_1}
Malicious 8-states QFactory satisfies $8$-state basis blindness with respect to the ideal set of states $\{\Ket{+}, \Ket{+_{\frac{\pi}{4}}}, .., \Ket{+_{\frac{7\pi}{4}}}\}$.
\end{theorem}
\begin{proof}

Firstly, we use \autoref{thm:security_qfac_2_0} to obtain that $B_1$ and $B_1'$ are completely hidden from any adversary. \\

We prove the following: \\
If there exists a QPT adversary $\mathcal{A}$ that is able to break the 8-states basis blindness property of Malicious 8-states QFactory (determine the indices $L_2$ and $L_3$ with probability $\frac{1}{4} + \frac{1}{poly_1(n)}$ for some polynomial function $poly_1$), then we can construct a QPT adversary $\mathcal{A}'$ that can break the 4-states basis blindness of the Malicious 4-states QFactory protocol (determine the basis bit with probability $\frac{1}{2} + \frac{1}{poly_2(n)}$, for some polynomial $poly_2(\cdot)$). \\

The input to $\mathcal{A'}$ should be consisting only of the $\mathcal{F}$ family index, $k$, and the description of $h$. Next we show how to construct $\mathcal{A'}$ to determine the corresponding index $B_1$ or $B_1'$ of the output state (of one of the 2 Malicious 4-states QFactory runs), by using as a subroutine $\mathcal{A}$ that acts as follows: receives as input 2 function indices $k^{(1)}$ and $k^{(2)}$, runs Malicious 8-states QFactory and then $\mathcal{A}$ is able to output the correct basis bits $L_2$ and $L_3$, with probability $\frac{1}{4} + \frac{1}{poly_1(n)}$.

Before we describe $\mathcal{A}'$, we need to define the following 3 values: \\
$P_2$ = probability that $\mathcal{A}$ guesses correctly $L_2$; \\
$P_3$ = probability that $\mathcal{A}$ guesses correctly $L_3$; \\
$P_{\oplus}$ = probability that $\mathcal{A}$ guesses correctly $L_2 \oplus L_3$; \\

Now, given that $\mathcal{A}$ is able to produce both $L_2$ and $L_3$ with probability $\frac{1}{4} + \frac{1}{poly_1(n)}$, this implies that $max(P_2, P_3, P_{\oplus}) \geq \frac{1}{2} + \frac{1}{poly_2(n)}$ for some polynomial $poly_2(\cdot)$ (see proof in \autoref{app:simple_lemmata}). \\
We will construct $\mathcal{A}'$ such that if $P_3$ is the maximum, then $\mathcal{A}'$ can determine $B_1$ (break the basis blidness of the first Malicious 4-states QFactory run) and if  $P_2$ or $P_{\oplus}$ is the maximum, then $\mathcal{A}'$ can determine $B_1'$ (break the basis blidness of the the second ''rotated'' Malicious 4-states QFactory run). \\

\procedure [linenumbering]{$\mathcal{A'}(k, h, 1^n)$} {
\pccomment{$k^{(1)}$ corresponds to the input for the first run of Malicious 4-states QFactory - with the output index $(B_1, B_2)$, } \\
\pccomment{while $k$ corresponds to the input for the second run - with the output index $(B_1', B_2')$} \\
(k^{(1)}, t_k^{(1)}) \sample {Gen}_{\mathcal{F}}(1^n) \\
\pccomment{As the probability P of successfully guessing $L_2$ and $L_3$ is $\frac{1}{4} + \frac{1}{poly_1(n)}$ } \\
\pccomment{We know that $max(P_2, P_3, P_{\oplus}) \geq \frac{1}{2} + \frac{1}{poly_2(n)}$}\\
\pcif P_3 = max(P_2, P_3, P_{\oplus}) \\
\pccomment{we break the basis-blindness of the first Malicious 4-states QFactory by determining $B_1$} \\
\pcind (y^{(1)}, y^{(2)}, b^{(1)}, b^{(2)}, s_1, s_2), (\tilde{L_2}, \tilde{L_3}) \gets \mathcal{A}(k, k^{(1)}, h) \\
\pcind \pccomment{$(y^{(1)}, y^{(2)}, b^{(1)}, b^{(2)}, s_1, s_2)$ represents the classical communication received from $\mathcal{A}$} \\
\pcind \pccomment{during the run of Malicious 8-states QFactory, and $(\tilde{L_2}, \tilde{L_3})$ - are the guesses of $\mathcal{A}$ for the indices of the outcome} \\
\pcind \tilde{B_1} \gets \tilde{L_3} \\
\pcind \pcreturn \tilde{B_1} \pccomment{as $B_1 = 1 \oplus L_3$ as seen in Eq. \ref{eq:expressions_l_3} and we have success probability $\geq \frac{1}{2} + \frac{1}{poly_2(n)}$}\\
\pcelse \\
\pccomment{we break the basis-blindness of the second Malicious 4-states QFactory by determining $B_1'$} \\
\pcind (y^{(1)}, y^{(2)}, b^{(1)}, b^{(2)}, s_1, s_2), (\tilde{L_2}, \tilde{L_3}) \gets \mathcal{A}(k^{(1)}, k, h) \\
\pcind \pccomment{$(y^{(1)}, y^{(2)}, b^{(1)}, b^{(2)}, s_1, s_2)$ represents the classical communication received from $\mathcal{A}$} \\
\pcind \pccomment{during the run of Malicious 8-states QFactory, and $(\tilde{L_2}, \tilde{L_3})$ - are the guesses of $\mathcal{A}$ for the indices of the outcome} \\
\pcind (z^{(1)}, z'^{(1)}) \gets Inv_{\mathcal{F}}(y^{(1)}, t_k^{(1)}) \\
\pcind B_1 \gets h(z^{(1)}) \oplus h(z'^{(1)}) \\
\pcind B_2 \gets [b^{(2)}_n + \sum (z^{(2)}_i \oplus z'^{(2)}_i) \cdot b^{(2)}_i] B_1 + h(z^{(2)})(1 - B_1) \bmod 2  \\
\pcind \pcif P_2 = max(P_2, P_3, P_{\oplus}) \\
\pcind \pcind \pccomment{Then $B_1' =  L_2 \oplus B_1 \cdot (B_2 \oplus s_2)$ as seen in Eq. \ref{eq:expressions_l_2}} \\
\pcind \pcind \tilde{B_1'} \gets \tilde{L_2} \oplus [B_1 \cdot (B_2 \oplus s_2)] \\
\pcind \pcind \pcreturn \tilde{B_1'} \\
\pcind \pcif P_{\oplus} = max(P_2, P_3, P_{\oplus}) \\
\pcind \pcind \tilde{B_1'} \gets \tilde{L_2} \oplus \tilde{L_3} \oplus B_1 \oplus [B_1 \cdot (B_2 \oplus s_2)] \\
\pcind \pcind \pcreturn \tilde{B_1'} \\
}

\end{proof}

\section{Malicious-abort 4-states QFactory: treating the abort case}\label{sec:qfactory_abort}

In this section, we will discuss an extension of Malicious 4-states QFactory, whose aim is to achieve basis blindness even against adversaries that could try to exploit the fact that Malicious 4-states QFactory can abort when there is only one preimage associated to the $y$ returned by the server. One may think that we could just send back this $\accept/\abort$ bit to the server, but unfortunately it could leak additional information on the hardcore bit $d_0$ (which corresponds to the basis $B_1$ of the produced qubit) to the server, and from an information theory point of view, as soon as the probability of acceptance is small enough, we cannot guarantee that this bit remains secret. On the other hand, for honest servers, the probability of aborting is usually non-negligible, so we cannot neglect this case.

We stress out that it is also possible to guarantee that for honest servers this probability goes negligibly close to 1 by making an appropriate choice of parameters for the function, and in that case the initial protocol of Malicious QFactory defined \autoref{sec:noabortprotocol} that outputs a random state on Alice's side when $y$ does not have two preimages is secure, but this comes (as far as we know), at the cost of using a function with is ``less'' secure. More specifically, instead of having a reduction to $\GapSVP$ with a polynomial $\gamma$, the reduction usually goes to $\GapSVP$ with a super-polynomial $\gamma$. Such function parameters has been used implicitly in other works \cite{urmila_qfhe} (\cite{brakerski_qfhe} later removed this assumption), and for now they are believed to be secure (the best known polynomial algorithm cannot break $\GapSVP$ with a $\gamma$ smaller than exponential), but these assumptions are usually not widely accepted in the cryptography community, and that's why we aim to remove this non-standard assumption.

The solution we propose in this section makes the assumption that the classical Yao's XOR Lemma also applies for one-round protocols (with classical messages) against quantum adversary. This lemma roughly states that if you cannot guess the output bit of one round with probability better than $\eta$, then it's hard to guess the output bit of $t$ independents rounds with probability much better than $1/2 + \eta^t$. As far as we know, this lemma as been proven only in the classical case (see \cite{Goldreich2011} for a review of this theorem as well as the main proof methods), and some works \cite{VW08} even extend this lemma to protocols, and also to the quantum setting \cite{Sherstov2010, Klauck2004}. Unfortunately these last works focus mostly on communication and query complexity, and are not really usable in our case.

Note also that we are also working on some other proof methods to get rid of this last assumption and to improve the efficiency of the protocol by avoiding repetition (see \autoref{sec:discussion_qfactory_abort} for more details).

In the following, we will call ``accepted run'' a run of Malicious 4-states QFactory such that the received $y$ from the server has 2 preimages (``probability of success'' also refers to the probability of this event when the server is honest), and otherwise we call it an ``aborted run''.

\subsection{The Malicious-Abort 4-state QFactory Protocol}

In a nutshell, the solution we propose is to run several instances of Malicious 4-states QFactory, by remarking that we do not need to discard the aborted runs. Indeed, it is easy to see that in these cases, the produced qubits will always be in the same basis (denoted by 0). The idea is then to implement on the server side a circuit that will output a qubit having as basis the XOR of all the basis of the accepted runs (without even leaking which runs are accepted or not), and check on client's side that the number of accepted runs is high enough (this will happen with probability exponentially close to 1 for honest servers). If it is the case, the client will just output the XOR of the basis of the accepted run, and otherwise (i.e. if the server is malicious), he will just pick a random bit value.

Unfortunately, in practice things are a bit more complicated in order to be able to write the proof of security, and we need to divide all the $t$ runs into $n_c$ ``chunks'' of size $t_c$, and test them individually. Here is a more precise (but still high level) description of the protocol and proof's ideas:
\begin{itemize}
\item firstly, we run $t = n_c \times t_c$ parallel instances of Malicious 4-states QFactory, without revealing the abort bit for any of these instances;
\item then the key point to note is that for honest servers, if $y_i$ has only one preimage then the output qubit produced by the server at the end of the protocol will be either $\ket{0}$ or $\ket{1}$, but cannot be $\ket{+}$ or $\ket{-}$ (with one preimage we do not have a superposition). In other words, the basis is always the $\{\ket{0},\ket{1}\}$ basis (denoted as $0$) so we do not really need to abort. Therefore, at the end, (for honest runs) the basis of the output qubits will be equal for all $i \in [\![0,t]\!]$ to $\beta_i = d_{0,i} \cdot a_i$, where $a_i = 1$ iff $y_i$ has two preimages, and $a_i = 0$ otherwise. Of course, this distribution will be biased against $0$, but it is not a problem. See \autoref{lem:aborted_useful} for proof.
\item then, it also appears that from $t$ qubits in the basis $\beta_1,\dots, \beta_t$, we have a way to produce a single qubit belonging to the set $\{\ket{0}, \ket{1}, \ket{+}, \ket{-}\}$ whose basis $B_1$ is the XOR of the basis of the $t$ qubits, i.e. $B_1 = \oplus_{i=1}^t \beta_i$ (see \autoref{lem:circuit_gadget_xor}).
\item Then, the client will test every chunk, by checking if the proportion of accepted runs in every chunk is greater than a given value $p_c$. If all chunks have enough accepted runs, then the client just computes and outputs the good value for the basis (which is the XOR of the hardcore bit of all the accepted runs) and value bits. However, if at least one chunk does not have enough accepted runs (which shouldn't happen if the server is honest), then the client just outputs random values for the basis and value bit, not correlated with server's qubit (it's equivalent to say that a malicious server can always through away the qubit and pick a new qubit not correlated with client's one).
\item Correctness: if the probability to have two preimages for an honest server is at least a constant $p_a$ greater than $1/2$ (the parameters we proposed in \cite{qfactory_old} have this property), and if $t$ is chosen high enough, the fraction of accepted runs will be close to $p_a$, and we can show that the probability to have a fraction of accepted runs smaller than a given constant $p_b < p_a$ is exponentially (in $t$) close to 0 (cf \autoref{lem:proba_success_honest_factory2.2_chunk}). So with overwhelming probability, all the chunks will have enough accepted runs, i.e. honest servers will have a qubit corresponding to the output of the client.
\item Soundness: to prove the security of this scheme, we first prove \autoref{lem:one_chunk_pretty_diff} that it is impossible for any adversary to guess the output of one chunk with a probability bigger than a constant $\eta < 1$ (otherwise we have a direct reduction that breaks the hardcore bit property of $g_K$). Now, using the quantum version of Yao's XOR Lemma that we conjecture at \autoref{conj:quantumyaoxorlemma}, we can deduce that no malicious server is able to guess the XOR of the $t_c$ chunks/instances with probability better than $1/2 + \eta^{t_c}+\negl$, which goes negligibly close to $1/2$ when $t_c = \Omega(n)$.
\end{itemize}
So putting everything together, the parties will just run $t = n_c \cdot t_c$ Malicious 4-states QFactory in parallel, the client will then check if $\sum_i a_i$ is higher than $p_c \cdot t_c$ for all the $n_c$ chunks, and if so he will set $B_1 = \oplus_{i = 1}^{t} d_i \cdot a_i$ (server has a circuit to produce a qubit in this basis as well). Otherwise $B_1$ will be set to a uniformly chosen random bit (it is equivalent to say that a malicious server can destroy the qubit, and this is also unavoidable even with a real quantum communication), and we still have correctness with overwhelming probability for honest clients.

For the exact algorithm see \autoref{protocol:qfactory_abort_real}, and go to \autoref{thm:malicious_abort_secure_correct} for the theorem/proof of security (some of the proofs have been moved to appendix \autoref{sec:proof_malicious_appendix}).

\subsection{Correctness and security of Malicious-Abort 4-state QFactory}
Now, we will formalize and prove the previous statements.

\begin{conjecture}[Yao's XOR Lemma for one-round protocols (classical messages) against quantum adversary]\label{conj:quantumyaoxorlemma}~\\
  Let $n$ be the security parameter, let $f_n: \cX_n \times \cY_n \rightarrow \{0,1\}$ be a (possibly non-deterministic) family of functions (usually not computable in polynomial time), and let $\chi_n$ be a distribution on $\cX_n$ efficiently samplable. If there exists $\delta(n)$ such that $|\delta(n)| \geq \frac{1}{\poly}$ and such that for all polynomial (in $n$) quantum adversary $\cA_n: \cX_n \rightarrow \cY_n \times \{0,1\}$,
  \[\Pr\left[ \tilde{\beta} = f_n(x,y) \mid (y,\tilde{\beta}) \leftarrow \cA_n(x), x \leftarrow \chi_n\right] \leq 1 - \delta(n)\]
  then, for all $t \in \N^*$, there is no polynomial quantum adversaries $\cA_n': \cX_n^{t} \rightarrow \cY_n^{t} \times \{0,1\}$ such that
  \[\Pr\left[ \tilde{\beta} = \bigoplus_{i=1}^t f_n(x_i,y_i) \mid (y_1,\dots,y_t,\tilde{\beta}) \leftarrow \cA_n'(x_1,\dots,x_t), x_1 \leftarrow \chi_n, \dots, x_t \leftarrow \chi_n\right] \geq \frac{1}{2} + (1 - \delta(n))^t + \negl\]
\end{conjecture}

\begin{lemma}[Aborted runs are useful]\label{lem:aborted_useful}
  If $\piAFourWeak$ and $\piBFourWeak$ are following the Malicious 4-states QFactory protocol honestly, and if $y$ has not 2 preimages, then the output qubit produced by $\piBFourWeak$ is in the basis $\{\ket{0}, \ket{1}\}$.
\end{lemma}
\begin{proofatend}
  The function $f_k$ cannot have more than two preimages by assumption, and in the Malicious 4-states QFactory protocol the output $y$ is in the image of $f_k$. So it means that $y$ has exactly one preimage $x$. So after measuring the last register, the states will be in the state $\ket{0} \otimes \ket{x} \otimes \ket{y}$. Then, we apply $U_h$, so the states becomes $\ket{d} \otimes \ket{x} \otimes \ket{y}$ with $d \in \{0,1\}$. We remark that the first qubit is not entangled with the measured qubits, so the output qubit will be $\ket{d}$, which is indeed in the basis $\{\ket{0}, \ket{1}\}$.
\end{proofatend}

\begin{figure}[ht!]
  \[
    \Qcircuit{
 & \lstick{\ket{+_{\pi/2}}} & \ctrl{1} & \qw      & \qw      & \qw      & \qw      & \qw & \qw                 & \measure{\ket{\pm}} & \cw \, \, \,  \, \, \,  \, \, \,  s_{1,1} \\
 & \lstick{\ket{\quin_1}}   & \gate{Z} & \qw      & \qw      & \ctrl{6} & \qw      & \qw & \qw                 & \measure{\ket{\pm}} & \cw \, \, \,  \, \, \,  \, \, \,  s_{1,2} \\
 & \lstick{\ket{+_{\pi/2}}} & \qw      & \ctrl{1} & \qw      & \qw      & \qw      & \qw & \qw                 & \measure{\ket{\pm}} & \cw \, \, \,  \, \, \,  \, \, \,  s_{2,1}                       \\
 & \lstick{\ket{\quin_2}}   & \qw      & \gate{Z} & \qw      & \qw      & \ctrl{4} & \qw & \qw                 & \measure{\ket{\pm}} & \cw \, \, \,  \, \, \,  \, \, \,  s_{2,2}                       \\
 & \lstick{\vdots\quad}     &          &          & \vdots   &          &          & \vdots                                                                                  \\
 & \lstick{\ket{+_{\pi/2}}} &          &          & \ctrl{1} & \qw      & \qw      & \qw & \qw                 & \measure{\ket{\pm}} & \cw \, \, \,  \, \, \,  \, \, \,  s_{t,1} \\
 & \lstick{\ket{\quin_t}}   &          &          & \gate{Z} & \qw      & \qw      & \qw & \ctrl{1}            & \measure{\ket{\pm}} & \cw \, \, \,  \, \, \,  \, \, \,  s_{t,2 }\\
 & \lstick{\ket{+}}         & \qw      & \qw      & \qw      & \gate{Z} & \gate{Z} & \qw & \gate{Z}            & \qw \, \, \,  \, \, \,  \, \, \, \, \, \, \, \, \, \,  \ket{\quout}
    }
  \]
  \caption{The XOR gadget circuit $\gadxor$ (run on server side)}
  \label{fig:circuit_gadget_xor}
\end{figure}
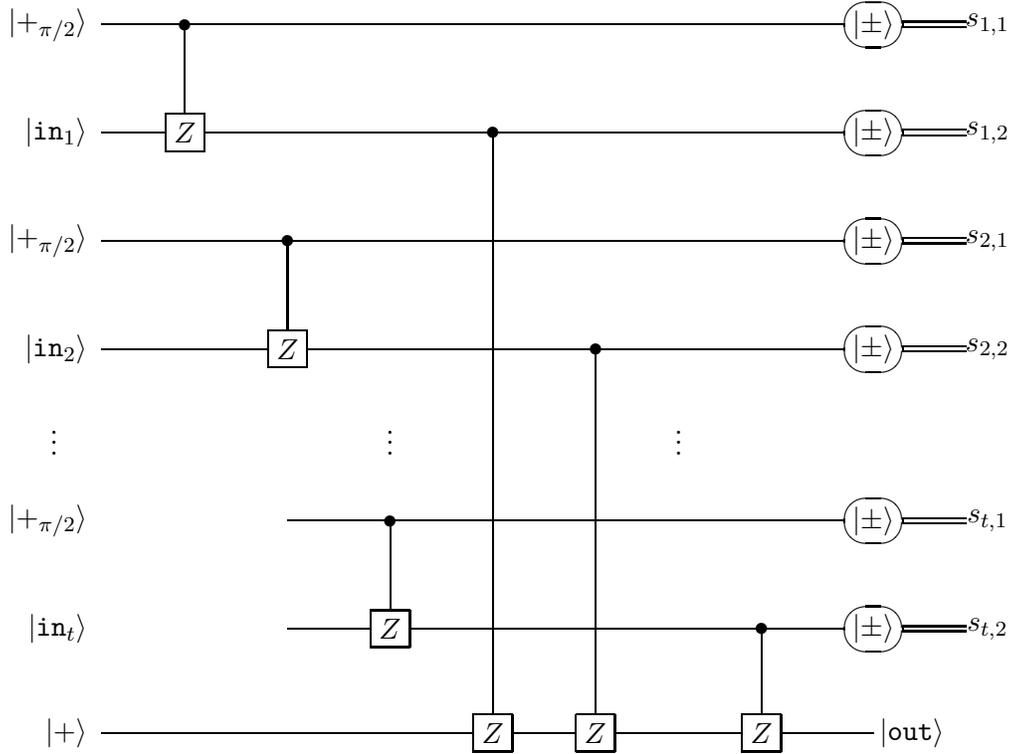

\begin{lemma}[Gadget circuit $\gadxor$ computes XOR]\label{lem:circuit_gadget_xor}
  If we denote by $b_i$ the basis of $\ket{\quin_i}$ (equal to $0$ if the basis is $0/1$, and 1 if the basis is $+/-$), and if we run the circuit $\gadxor$ (inspired by measurement based quantum computing) represented \autoref{fig:circuit_gadget_xor} on these inputs, then basis of $\ket{\quout}$ is equal to $\oplus_{i=1}^t b_i$.
\end{lemma}
\begin{proofatend}
The entire analysis of the circuit will be performed only with respect to the basis of the states of the circuit. Let us first examine the first part of the circuit, where we apply $\wedge Z$ between $\ket{+_{\frac{\pi}{2}}}$ and $\ket{\quin_1} = H^{{B_1}^{(1)}}Z^{{B_2}^{(1)}}$ (with ${B_1}^{(1)}$ the basis of $\ket{\quin_1}$) and then measure the first qubit in the $\ket{\pm}$ basis, and we denote the result state $V_1$. \\
The result of this operation is: \\
- if ${B_1}^{(1)} = 0$, $V_1 = R(\pi(B_2^{(1)} + s_{1, 1} + 1)) \ket{+_{\frac{\pi}{2}}} \in \{ \ket{+_{\frac{\pi}{2}}}, \ket{-_{\frac{\pi}{2}}} \}$ \\
- if ${B_1}^{(1)} = 1$, $V_1 = X^{B_2^{(1)}}\ket{0} \in \{ \ket{0}, \ket{1} \}$ \\
In other words, the state $V_1$ belongs to the basis $\mathcal{B}_0 = \{ \ket{+_{\frac{\pi}{2}}}, \ket{-_{\frac{\pi}{2}}} \}$, if ${B_1}^{(1)} = 0$ and to the basis $\mathcal{B}_1 = \{ \ket{0}, \ket{1} \}$ if ${B_1}^{(1)} = 1$. \\
Now, we can think of the circuit as having $t$ states $V_i \in \{\ket{0}, \ket{1}, \ket{+_{\frac{\pi}{2}}}, \ket{-_{\frac{\pi}{2}}}\}$, where every $V_i$ has the basis ${B_1}^{(i)}$. Then, to compute the output state $\ket{\quout}$ of $\gadxor$, for every $i \in \{1,...,t\}$ we have to apply $CZ$ between $V_i$ and $\ket{+}$ and then measure the first qubit in the $\ket{\pm}$ basis. \\
So let us do this step first for $V_1$. The result is a state $W_1 = X^{s_{1,2}}HV_1$, thus we obtain that: \\
$W_1$ belongs to the basis $\mathcal{B}_0 = \{ \ket{+_{\frac{\pi}{2}}}, \ket{-_{\frac{\pi}{2}}} \}$, if ${B_1}^{(1)} = 0$ and to the basis $\mathcal{B}_2 = \{ \ket{+}, \ket{-} \}$ if ${B_1}^{(1)} = 1$. \\
Next we do the same operations between $V_2$ and $W_1$, the result being a state $W_2$, then between $V_3$ and $W_2$ and so on, therefore, the outcome state is $\ket{\quout} = W_{t}$.\\
We will prove by induction that the state $W_{t} \in \{ \ket{+}, \ket{-}, \ket{+_{\frac{\pi}{2}}}, \ket{-_{\frac{\pi}{2}}} \}$, where the basis of $W_{t}$ is given by $B_1 = {B_1}^{(1)} \oplus ... \oplus {B_1}^{(t)} $. \\
As we have proved already for the basis case $t = 1$, we now prove the induction step.
Suppose that $W_{n} \in \{ \ket{+}, \ket{-}, \ket{+_{\frac{\pi}{2}}}, \ket{-_{\frac{\pi}{2}}} \}$ with basis $B_1 = {B_1}^{(1)} \oplus ... \oplus {B_1}^{(n)}$. \\
To obtain $W_{n+1}$ we have to apply $\wedge Z$ between $V_{n + 1}$ and $W_n$ and then measure the first qubit. Then after computing this, we obtain that the basis of $W_{n+1}$ is $B_1$ if the basis of $V_{n + 1}$ is ${B_1}^{(n + 1)} = 0$ and the basis of $W_{n+1}$ is $1 \oplus B_1$ if the basis of $V_{n + 1}$ is ${B_1}^{(n + 1)} = 1$. In other words, the basis of $W_{n+1}$ is given by $B_1 = {B_1}^{(1)} \oplus ... \oplus {B_1}^{(n)} \oplus {B_1}^{(n + 1)}$, which concludes the proof.
\end{proofatend}

We will now describe here the protocol of Malicious-Abort 4-states QFactory:
\begin{algorithm}[H]
  \caption{Malicious-Abort 4-states QFactory Protocol}\label{protocol:qfactory_abort_real}
  \textbf{Requirements:} \\
  Public: The family of functions $\cF: \cD \rightarrow \cR$ and $h$ described above, such that the probability of having two pre-images for a random image is greater than a constant $p_a > 1/2$.\\
  This protocol is based on some constants $t_c \in \N$ (number of repetitions per chunk), $n_c \in \N$ (number of chunks), $p_a \in (1/2,1]$ (lower bound on the probability to have an accepted run in the honest protocol), $p_c \in (1/2,1] < p_a$ (fractions of the runs per chunk that must be accepted). These constants can be chosen to have overwhelming probability of success for honest players, and negligible advantage for a malicious adversaries trying to guess the basis.\\
\textbf{Stage 1: Run multiple QFactories} \\
-- Client: prepares $t = n_c \times t_c$ public keys/trapdoors: $\left((k^{(i,j)},t_{k^{(i,j)}}) \leftarrow \text{Gen}_{\mathcal{F}}(1^n)\right)_{i \in [\![1,n_c]\!], j \in [\![1,t_c]\!]}$. The Client then sends the public keys $k^{(i,j)}$ to the Server, together with $h$.\\
-- Server and Client: follow \autoref{protocol:qfactory_real} $t$ times, with the keys sent at the step before. Client receives $((y^{(i,j)},b^{(i,j)}))_{i,j}$, and sets for all $i,j$: $a^{(i,j)} = 1$ iff $|f^{-1}(y^{(i,j)})|=2$, otherwise $a^{(i,j)} = 0$, and $B^{(i,j)}_1$ and $B^{(i,j)}_2$ like in \autoref{protocol:qfactory_real} when $a^{(i,j)} = 1$ (otherwise $B^{(i,j)}_1=0$ and $B^{(i,j)}_2=h(f^{-1}(y))$. Server will get $t$ outputs $\ket{\quin_{(i,j)}}$.\\
\textbf{Stage 2: Combine runs and output}\\
-- Server: applies circuit \autoref{fig:circuit_gadget_xor} on the $t$ outputs $\ket{\quin_t}$, and outputs $\ket{\quout}$.\\
-- Client: checks that for all chunks $i \in [\![1,n_c]\!]$ the number of accepted runs is high enough, i.e. $\sum_j a^{(i,j)} \geq p_c t_c$.
\begin{itemize}
\item If at least one chunk does not respect this condition, then picks two random bits $B_1$ (the basis bit) and $B_2$ (the value bit) and outputs $(B_1, B_2)$, corresponding to the description of the BB84 state $H^{B_1}X^{B_2}\ket{0}$.
\item If all chunks respect this condition, then sets $B_1 := \bigoplus_{i,j} B^{(i,j)}_1$ (the final basis is the XOR of all the basis), and $B_2$ will be chosen to match the output of \autoref{fig:circuit_gadget_xor}.
\end{itemize}
\end{algorithm}

\begin{lemma}[Probability of correctness of Malicious-Abort 4-states QFactory for one chunk]\label{lem:proba_success_honest_factory2.2_chunk}
  If the probability to have an accepted run in Malicious 4-states QFactory with honest parties is greater than a constant $p_a > 1/2$, i.e.
  \[\Pr[|f^{-1}_k(y)| = 2 \mid (\piAFourWeak \| \piBFourWeak) ] \geq p_a\]
  (where $\piAFourWeak$ and $\piBFourWeak$ are the honest protocols of Malicious 4-states QFactory) then the probability to have at least $p_b t_c$ accepted runs (with $p_b < p_a$, $p_b$ considered as a constant) is exponentially (in $t_c$) close to 1:
  \[\Pr\left[ \sum_i a_i \geq p_b t_c \mid (\piAFourXORchunk^{t_c} \| \piBFourXORchunk^{t_c})\right] \geq 1 - \frac{1}{e^{2 (p_a - p_b)^2 t_c }}= 1 - \negl[t_c]\]
  (where $\piAFourXORchunk^{t_c}$ and $\piBFourXORchunk^{t_c}$ are the (honest) parties of the \autoref{protocol:qfactory_abort_real} restricted on one chunk of size $t_c$, or, equivalently $t_c$ parallels repetition of the protocol \autoref{protocol:qfactory_real})
\end{lemma}
\begin{proofatend}
  In the honest case, all runs are independents, so let us define $\{A_i\}_{i=1}^t$ as the (binary) random variables whose values are 1 iff the $i$-th run has two preimages associated with $y_i$. We know that for all $i$, $\E(A_i) \geq p_a > p_b$. So let us define $\eps = \E(A_i) - p_b > p_a - p_b$. Using Chernoff inequality we have
  \[
    \Pr\left[\frac{1}{t} \sum_{i=1}^t A_i < \E(A_i) - \eps\right]
    \leq e^{-2\eps^2 t}
    \leq e^{-2 (p_a - p_b)^2 t}
    =\negl[t]
  \]
  (because $p_a-p_b$ is constant)
\end{proofatend}

\begin{lemma}[Correctness of \autoref{protocol:qfactory_abort_real}]\label{lem:correctness_qfactory_abort}
  The protocol \autoref{protocol:qfactory_abort_real} is correct with overwhelming probability as soon as $t = \poly[n]$ and $t_c = \Omega(n)$, i.e.
  \[\Pr\left[  \ket{\quout} = H^{B_1} Z^{B_2} \mid ((B_1,B_2), \ket{\quout}) \leftarrow (\pi_A \| \pi_B) \right] \geq 1 - \negl\]
\end{lemma}
\begin{proofatend}
  The \autoref{lem:proba_success_honest_factory2.2_chunk} gives that the probability to have more than $p_c t_c$ accepted runs for a given chunk is $1-\negl[t_c]$, i.e. if $t_c = \Omega(n)$, this probability is $\negl[n]$. So for $n_c$ chunks, the probability to have one fail is $(1-\negl[n])^{n_c} = 1 - \negl[n]$ as soon as $n_c = \poly[n]$, which is the case because $t = t_c \times n_c = \poly[n]$. Then, when all the chunks are accepted, the correctness of the output values is assured by \autoref{lem:circuit_gadget_xor}.
\end{proofatend}

\begin{definition}
  For all public key $k$ and image $y$, we define $a(k,y) = 1$ iff $|f_k^{-1}(y)| = 2$, and $a(k,y) = 0$ otherwise.\\
  Then, for all $t_c \in \N$ and $p_c \in [0,1]$, we define $\beta_{t_c,p_c}(k^{(1)},\dots,k^{(t_c)}, y^{(1)},\dots,y^{(t_c)})$ as the (randomized) function that outputs a random bit if $\sum_i a(k^{(i)}, y^{(i)}) < p_c \cdot t_c$, and outputs otherwise $\xor_i (a(k^{(i)}, y^{(i)}) \cdot d^{(i)}_{0})$, where $d^{(i)}_0$ is the hardcore bit corresponding to $k^{(i)} := (K^{(i)},g_{K^{(i)}}(z^{(i)}_0))$, i.e. $d^{(i)}_0 = h(z^{(i)}_0)$.
\end{definition}

\begin{lemma}[Solving one chunk is pretty difficult]\label{lem:one_chunk_pretty_diff}
  Let $p_c \in (\frac{1}{2}, 1]$. Then, there exists no polynomial adversary $\cA$ such that
  \[\Pr\left[ \tilde{B}_1 = \beta_{t_c, p_c}(k^{(1)},\dots,k^{(t_c)}, y^{(1)},\dots,y^{(t_c)}) \mid (y^{(1)},\dots,y^{(t_c)},\tilde{B}_1) \leftarrow \cA(k^{(1)},\dots,k^{(t_c)}) \right] > \eta\]
  with $\eta = \frac{1}{2} \left( 1 + \frac{1}{2p_c} \right)$, and the randomness being on the randomness of $\beta$, $\cA$, and on the choice of $(k^{(i)})_i$ and $(y^{(i)})_i$.
\end{lemma}
\begin{proofatend}
  By contradiction, let us assume that there is an adversary $\cA$ such that (we omit the parameters for readability)
  \[\Pr\left[ \tilde{B}_1 = \beta \right] > \eta\]
  Then, if we define $a_i := a(k^{(i)}, y^{(i)})$,
  \begin{align*}
    \eta
    &< \pr{\tilde{B}_1 = \beta}\\
    &= \underbrace{\pr{ \sum_i a_i < p_c t_c}}_\alpha \times \frac{1}{2} + \pr{ \sum_i a_i \geq p_c t_c} \times \pr{\tilde{B}_1 = \beta \mid \sum_i a_i \geq p_c t_c}\\
    &= \alpha \times \frac{1}{2} + (1-\alpha) \times \pr{\tilde{B}_1 = \beta \mid \sum_i a_i \geq p_c t_c}\\
    &\leq \alpha \times \frac{1}{2} + (1-\alpha) = 1 - \frac{\alpha}{2} \\
  \end{align*}
  so $\alpha \leq 2(1-\eta)$.

  Now, we remark that we can bound also $(1-a) \times \pr{\tilde{B}_1 = \beta \mid \sum_i a_i \geq p_c t_c}$. Indeed, if this value is too big then we can construct an adversary that could break the hardcore bit property of $g_K$. To do that, we define an adversary $\cA'$ taking as input a $k$, and whose goal is to define the hardcore bit $d_0$ associated with $k$. This adversary will pick $t_c - 1$ public keys/trapdoors $(k^{(i)},t_{k^{(i)}})$, and hide $k$ in the middle of these trapdoors. Then, $\cA'$ calls $\cA$ with these $t_c$ keys, and outputs $\tilde{d_0} := \tilde{B}_1 \xor_i a^{(i)}d_0^{(i)}$, with $\tilde{B}_1$ the output of $\cA$, and $a^{(i)}$ computed by using the $y^{(i)}$ provided by $\cA$. We know that $\tilde{d_0} = d_0$ when the guess of $\cA'$ was right, when $\sum_i a_i \geq p_c t_c$, and when the $y$ corresponding to the function $k$ has two preimages. But this even occurs with probability greater than $(1-\alpha) \times \pr{\tilde{B}_1 = \beta \mid \sum_i a_i \geq p_c t_c} \times p_c$, and because $d_0$ is a hardcore bit, this probability is bounded by $1/2 + \negl$, or equivalently:
  \begin{align*}
    (1-\alpha) \times \pr{\tilde{B}_1 = \beta \mid \sum_i a_i \geq p_c t_c}&\leq \frac{1}{2 p_c} + \negl
  \end{align*}
  Now, let's come back to our probability to guess $\beta$:
  \begin{align*}
    \pr{\tilde{B}_1 = \beta}
    &= \alpha \times \frac{1}{2} + (1-\alpha)\times \pr{\tilde{B}_1 = \beta \mid \sum_i a_i \geq p_c t_c}\\
    &\leq \alpha \times \frac{1}{2} + \frac{1}{2 p_c} + \negl\\
    &\leq 1-\eta + \frac{1}{2p_c} + \negl
  \end{align*}
  But on the other side, $\pr{\tilde{B}_1 = \beta} > \eta$, so
  \begin{align*}
    \eta &< 1-\eta + \frac{1}{2p_c} + \negl\\
    \eta &< \frac{1}{2} \left( 1 + \frac{1}{2p_c} \right) + \negl
  \end{align*}
  Because $\eta$ and $p_c$ are constants\footnote{note that if we give them a dependence on $n$, we can make sure that $\eta - \frac{1}{2} \left( 1 + \frac{1}{2p_c} \right)$ is non negligible, but for simplicity we will keep them constant} that do not depend on $n$, this equality is also true without the $\negl$:
  \begin{align*}
    \eta &< \frac{1}{2} \left( 1 + \frac{1}{2p_c} \right)
  \end{align*}
  which is absurd because $\eta = \frac{1}{2} \left( 1 + \frac{1}{2p_c} \right)$.
\end{proofatend}

\begin{theorem}[Malicious-Abort QFactory is correct and secure]\label{thm:malicious_abort_secure_correct}
  Assuming \autoref{conj:quantumyaoxorlemma}, and by making sure that the probability for the family $\cF$ to have two preimages for a random image is bigger than a constant $p_a > 1/2$, then there exists a set of parameters $p_c$, $t_c$ and $n_c$ such that  \autoref{protocol:qfactory_abort_real} is correct with probability exponentially close to $1$ and basis-blind, i.e. such that for all polynomial adversaries $\cA$:
  \[\pr{\tilde{B}_1 = B_1 \mid ((B_1,B_2),\tilde{B}_1) \leftarrow (\piAFourXOR \| \cA) } \leq \frac{1}{2} + \negl\]
  More precisely, we need $t_c \in (1/2, p_c)$ to be a constant, and both $t_c$ and $n_c$ need to be polynomial in $n$ and $\Omega(n)$.
\end{theorem}
\begin{proofatend}
  The proof of correctness is made \autoref{lem:correctness_qfactory_abort}, and the security is a direct application of \autoref{conj:quantumyaoxorlemma}: after using \autoref{lem:one_chunk_pretty_diff}: this theorem provides a $\eta$ such that it's not possible to solve one chunk with probability better than $\eta < 1$, so $\delta(n) := 1 - \eta$ is a constant (and $\delta(n) \geq \frac{1}{\poly}$). Therefore \autoref{conj:quantumyaoxorlemma} tells us that no adversary can get the XOR of $n_c$ chunks with probability better than $\frac{1}{2} + \eta^{n_c} + \negl$. But $t_c = \Omega(n)$ and $\eta$ is a constant, so no adversary can get the XOR of $n_c$ chunks with probability better than $\frac{1}{2} + \negl$, i.e. no adversary can find $B_1$ with probability better then $\frac{1}{2} + \negl$.
\end{proofatend}

\section{Verifiable QFactory \label{sec:verif}}

In the preceding protocols, Malicious 4-states QFactory and Malicious 8-states QFactory, the produced qubits came with the guarantee of \textit{basis-blindness} (\autoref{def:4basisblind} and \autoref{def:8basisblind}). While this property refers to the ability of a malicious adversary to guess the honest basis bit(s), it tells nothing about the actual state that a deviating server might produce. For a number of applications, and most notably for \textit{verifiable blind quantum computation} \cite{fk}, the basis-blindness property is not sufficient. What is required is a stronger property, \textit{verification}, that ensures that the produced state was essentially prepared correctly, even in a malicious run.

\subsection{Verifiable QFactory Functionality}

There are two issues with trying to define a verification property for QFactory. The first is that the adversarial server can always abort, therefore the verification property can only ensure that the probability of non-abort \emph{and} cheat is negligible. The second issue is that, since the final state is in the hands of the server, the server can always apply a final deviation on the state\footnote{However that deviation needs to be independent of anything that is secret.}. This is not different with what happens in protocols that do have quantum communication. In that case, the adversarial receiver (server) can also apply a deviation on the state received before using that state in any subsequent protocol. This deviation could even be that the server replaces the received state with a totally different state. 

Here, we define the strongest notion of verifiable QFactory possible, which exactly captures the idea of being able to recover the ideal state from the real state without any knowledge of the  
(secret) index of the ideal state. In \autoref{app:vQFactory_channel} we show that this notion is sufficient for any protocol that includes communication of random secret qubits of the form $\Ket{+_{L\pi/4}}$ which includes a verifiable quantum computation protocol. Furthermore, in \autoref{app:strong_blindness} we show that it is possible to relax slightly the definition of verifiable QFactory. 

\begin{definition}[Verifiable QFactory\label{def:verifiable_8_state}]
Consider a party that is given a state uniformly chosen from a set of eight states $S = \{\rho_L \, | \, L \in \{0,1, ..., 7\}\}$ or an abort bit, where $S$ is basis-blind i.e. given a state sampled uniformly at random from $S$, it is impossible to guess the last two bits of the index $L$ of the state within the set $S$ with non-negligible advantage. We say that this party has a Verifiable QFactory if, it aborts with small probability  and when he does not abort, there exists an isometry $\Phi$, that is independent of the index $L$, such that:

\EQ{\label{eq:isometry_vQFactory}
\Phi(\rho_L)\overset{\epsilon}{\approx}\ket{+_{L\pi/4}}\bra{+_{L\pi/4}}\otimes \sigma_{\textrm{junk}}
}
where the state $\sigma_{\textrm{junk}}$ is independent of the index $L$.
\end{definition}

It is worth stressing here, that if the security setting is computational (as in this work), the basis-blindness and the approximate equality above involve a QPT distinguisher, while the isometry $\Phi$ needs to be computable in polynomial time.

\subsection{Blind Self-Testing \label{sec:blind_self_testing}} 

Before giving a verifiable QFactory protocol, we define a new concept of blind self-testing, that will be essential in proving the security of the former. Self-testing is a technique developed \cite{self_testing_yao, self_testing_magniez, self_testing_dam, robust_self_testing_mckague, mckague_self_test_graph} that ensures that given some measurement statistics, classical parties can be certain that some untrusted quantum state (and operations), that two or more quantum parties share, is essentially (up to some isometry) the state that the parties believe they have. In high-level, we are going to use a test of this kind in order to certify that the output of Verifiable QFactory is indeed the desired one.

Existing results, that we will call \emph{non-local self-testing}, only deal with how to exploit the non-locality (the fact that the quantum state tested is shared between non-communicating parties) to test the state and operations. Naturally, the correctness is up to a local isometry (something that the servers can apply, while preserving the non-communication condition). 

Here, instead of testing a single non-local state, we test a family of states, where the \textit{non-locality} property is replaced by the \textit{blindness} property (the fact that the server is not aware - he is blind - of which state from the (possibly known) family of states he is actually given in each run of the protocol). To see how this is closely related, one can imagine the usual non-local self-testing of the singlet state, where one quantum side (Alice) actually performs a measurement (as instructed). From the point of view of the other quantum side (Bob), he has a single state that (in the honest run) is one of the BB84 states, while he is totally oblivious about the basis of the state he has (if that was not the case it would lead to signalling the basis choice of Alice's measurement). However, this is, by no means the most general case. Here we introduce the concept of \emph{blind self-testing} formalising the above intuition.

We give here the most general case of blind self-testing and we conjecture that it holds. In \autoref{app:blind_self_testing_general} we list three simpler scenarios (of increasing complication) that lead to the most general case given here, following similar steps with the extension of simple i.i.d. self-testing to fully robust and rigid self-testing in existing literature  \cite{self_testing_magniez, robust_self_testing_mckague, robust_self_testing_yang, rigidity_ruv, verifier_leash}. In \autoref{app:blind_self-testing_iid} we provide a proof for the security of the first (simpler) case, while the full analysis of the most general blind self-testing goes beyond the scope of this work (as indicated by the large volume of high profile papers in the non-local self-testing cases) and will be the topic of a future publication.

\begin{algorithm}[H]
\caption{Blind self-testing: The general case}\label{protocol:general_blind_self-testing}
-- The server prepares a single state $\rho_{tot}$ (consisting of $N$ qubits in the honest run). This state has a corresponding index consisting of $N$ 3-bit indices $L_i$ ($L_i \in \{0,...,7\} \, \forall i \in \{1,...,N\}$), and the server is basis blind with respect to each of these $N$ indices, i.e. (being computationally bounded) he cannot determine with non-negligible advantage the two basis bits of any index $L_i$. On the other hand, the client does know all the indices $L_i$'s.\\
-- The client, randomly chooses a fraction $f$ of the qubits to be used as tests and announces the set of corresponding indices $T=\{i_1,\cdots,i_{fN}\} \subset \{1,2,...,N\}$ to the server. \\
-- For each test qubit $i_j \in T$, the client chooses a random measurement index $M_{i_j}\in\{000,\cdots,111\}$ and instructs the server to measure the corresponding qubit in the $\left\{\ket{+_{M_{i_j}\pi/4}},\ket{-_{M_{i_j}\pi/4}}\right\}$ basis. \\
-- The server returns the test measurement results $\{c_{(i_j)}\}$. \\
--- For each fixed pair $(L,M)$, the client gathers all the test positions that correspond to that pair and 
from the relative frequencies, the client obtains an estimate for the probability $p_{L,M}$ (where by convention we have that $p_{L,M}$ corresponds to the $+1$ outcome, while $1-p_{L,M}$ to the $-1$). 
\\
-- If $|p_{L,M}-\cos^2((L-M)\pi/8)|\geq \epsilon_2$ for any pair $(L,M)$ the client aborts.\\
\textbf{Output:} If the client does not abort (and this happens with non-negligible probability), then there exists an index-independent polynomial isometry $\Phi=\Phi_{k_1}\otimes\cdots\Phi_{k_l}$, given by products of the isometries in Fig. \ref{fig:isometry_circuit}, that is applied to a random subset of non-tested qubits $i$, such that:

\EQ{\label{eq:general_self_testing_isometry}
\Phi(Tr_{\textrm{all but } k_1,\cdots,k_l \textrm{ qubits}}\rho_{tot})\overset{\epsilon(\epsilon_1,\epsilon_2)}{\approx} \left(\ket{+_{L_{k_1}\pi/4}}_{k_1}\otimes\cdots\otimes \ket{+_{L_{k_l}\pi/4}}_{k_l}\right)\otimes\sigma_{\textrm{junk}}
}
\end{algorithm}

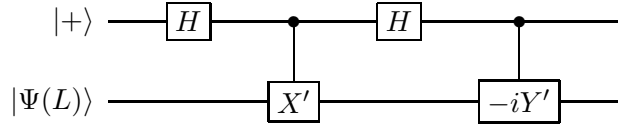
\begin{figure}[H]
\[
\Qcircuit @C=2em @R=1.4em {
   & \lstick{\ket{+}} & \gate{H} & \ctrl{1}  & \gate{H} & \ctrl{1} \qw & \qw & \\
   & \lstick{\ket{\Psi(L)}}  & \qw & \gate{X'}  & \qw & \gate{-iY'} & \qw  &
}
\]
  \caption{The isometry of the blind self-testing. Note that the controlled gates are controlled in the $X$-basis, i.e. $\wedge U_{12} (a\ket{+}+b\ket{-})_1\otimes (\ket{\psi})_2=a\ket{+}_1\otimes\ket{\psi}_2+b\ket{-}_1\otimes U\ket{\psi}_2$.}\label{fig:isometry_circuit}
\end{figure}

In this most general setting, we make no assumption on the state $\rho_{tot}$ produced by the server and moreover, we want to recover the full tensor product structure of the resulting states as given by Eq. (\ref{eq:general_self_testing_isometry}). In the non-local self-testing literature,  Azuma-Hoeffding, quantum de-Finetti theorems and rigidity results \cite{Hoeffding, azuma, de_finetti_caves, de_finetti_brandao} were used to uplift the simple i.i.d. case and prove the security in the most general setting.

\subsection{The Verifiable QFactory Protocol}

In this section we introduce a protocol for the final version of our functionality, Verifiable QFactory. 
Here, we give the protocol, show the correctness and the security, namely that the protocol achieves the verification property from \autoref{def:verifiable_8_state}, based on the conjectured security of the most general  \textit{blind self-testing} given in \autoref{protocol:general_blind_self-testing}. \\
The basic idea is the following: repeat the Malicious 8-states QFactory multiple times, then the client chooses a (random) fraction of the output qubits and uses them for a test and next instructs the server to measure the test qubits in random angles and, finally, the client checks their statistics. Since the server does not know the states (or to be more precise, the basis bits), he is unlikely to succeed in guessing the correct statistics unless he is honest (up to some trivial relabelling). Note that the output qubits and the measurement angles, need to be from all the set of the 8-states, which is one of the reasons we wanted to give the 8-states extension of our basic Malicious 4-states QFactory protocol\footnote{This is actually related with Bell's theorem as one can see later from the similarity with self-testing results.}. 

\begin{algorithm}[H]
\caption{Verifiable QFactory}\label{protocol:vQFactory}
\textbf{Requirements:} Same as in Protocol \autoref{protocol:qfactory_real}\\
\textbf{Input:} Client runs $N$ times the algorithm $(k^{(i)},t^{(i)}_k) \leftarrow \text{Gen}_{\mathcal{F}}(1^n)$, where $i\in\{1,\cdots, N\}$ denotes the $i$th run. He keeps the $t_k^{(i)}$'s private.\\
\textbf{Protocol:}\\
-- Client: runs $N$ times the Malicious 8-states QFactory \autoref{protocol:QFactory4to8}. 

-- Client: records measurement outcomes $y^{(i)},b^{(i)}$ and computes and stores the corresponding index of the output state $L^{(i)}$.

-- Client: instructs the server to measure a random fraction $rf$ of the output states, each in a randomly chosen basis of the form $\{\ket{+_{M^{(i)}\pi/4}},\ket{-_{M^{(i)}\pi/4}}\}$. Here $M^{(i)}$ is the index of the measurement instructed.

-- Server: returns the measurement outcomes $c^{(i)}$.

-- Client: for each pair $(L,M)$ collects the results $c^{(j)}$ for all $j$'s that have the specific pair and with the relative frequency obtains an estimate for the probability $p(L,M)$.

-- Client: aborts unless all the estimates of the probabilities $p(L,M)$ are $\epsilon$-close to the ideal one i.e. $p(L,M)\overset{\epsilon}{\approx}|\bra{+_{M\pi/4}}+_{L\pi/4}\rangle|^2$.

\textbf{Output:} The probability of non-aborting and being far from the ideal state\footnotemark \, is negligible $\epsilon'$

\EQ{
p(\textrm{non-abort}\wedge \Delta(\rho_{L^{(i_1)}\cdots L^{(i_{N(1-f)})}}, \rho_{\textrm{ideal}})\geq t(n))\leq\epsilon'
}
where $i_1, ..., i_{N(1-f)}$ refer to the unmeasured qubits 
 and where
\EQ{\label{eq:vQFactory_isometry_protocol}
\Phi(\rho_{\textrm{ideal}})=\otimes_{k=1}^{N(1-f)}\ket{+_{L^{(i_k)}\pi/4}}\bra{+_{L^{(i_k)}\pi/4}}\otimes \sigma_{\textrm{junk}}
}
and $\sigma_{\textrm{junk}}$ is a constant density matrix, while $\epsilon,\epsilon'$ are all negligible functions and $t(\cdot)$ is a non-negligible function.

Moreover, in an honest run, the probability of abort is negligible and the output is:

\EQ{
\rho_{\textrm{honest}}=\otimes_{k=1}^{N(1-f)}\ket{+_{L^{(i_k)}\pi/4}}\bra{+_{L^{(i_k)}\pi/4}}
} 
\end{algorithm}

\footnotetext{The distance $\Delta$ used here depends on the setting. In our case it is understood as a QPT distinguisher.}

\begin{theorem}[correctness]
If Protocol \ref{protocol:vQFactory} is run honestly, it aborts with negligible probability and the output (non-measured) qubits are exactly in a product state of the form $\ket{+_{L^{(i_k)}\pi/4}}\bra{+_{L^{(i_k)}\pi/4}}$. Therefore, the trivial isometry (the identity) suffices to recover the state of Eq. (\ref{eq:isometry_vQFactory}), and where there is no junk state.
\end{theorem}

\begin{proof}
In an honest run, each of the outputs of different Malicious 8-states QFactory runs, are of the correct form, therefore measuring any of those outputs in the $\{\ket{\pm_{M\pi/4}}\}$ basis returns the correct statistics with high probability. Hence, the protocol does not abort, while the remaining states are also prepared correctly.
\end{proof}

\begin{theorem}[security]
Protocol \ref{protocol:vQFactory} is a Verifiable QFactory (\autoref{def:verifiable_8_state}), i.e. the probability of accepting the tests and having a state far from the ideal is negligible irrespective of the deviation of the adversary, assuming that the self-testing \autoref{protocol:general_blind_self-testing} is correct.
\end{theorem}

\begin{proof}[Sketch of Proof.]

The outputs of the 8-states QFactory are basis blind, and satisfying the measurement statistics too, leading  exactly to the requirements of the definition of general blind self-testing in \autoref{protocol:general_blind_self-testing}. It follows that there exists an isometry such as that requested by Eq. (\ref{eq:vQFactory_isometry_protocol}).

\end{proof}
 
See also in \autoref{app:blind_self_testing_general} and \autoref{app:blind_self-testing_iid}, where
we explain how this task is very similar with self-testing results, and provide the first step for our self-testing result.

\section{Acknowledgements}

LC is very grateful to C\'{e}line Chevalier for all the discussions he had with her. He would also like to give a special thanks to Geoffroy Couteau, Omar Fawzi and Alain Passelègue who gave him great advices concerning security proof methods.
AC and PW are very grateful to Atul Mantri, Thomas Zacharias and Yiannis Tselekounis, Vedran Dunjko for very helpful and interesting discussions.
The work was supported by the following grants FA9550-17-1-0055 and EPSRC grants: EP/N003829/1 and EP/M013243/1. \\ \\
After completion of this work we became aware of the independent related work ``Computationally-secure and composable remote state preparation'' by Gheorghiu and Vidick.

\newpage

\begin{appendices} 
\section{Probability of guessing two predicates}\label{app:simple_lemmata}

\begin{lemma}[Implication of guessing two predicates]\label{lem:know_ab_a_b_or_xor}~\\
  Let $(a,b) \in \{0,1\}^2$ be two bits sampled uniformly at random. Let $f$ be any function of $(a,b)$ (eventually randomized). Then if $\cA$ is an adversary such that $\Pr[\cA(f(a,b)) = (a,b)] \geq 1/4 + \frac{1}{\poly}$ (where the probability is taken over the choice of $a$ and $b$, the randomness of $f$ and $\cA$), then either:
  \begin{itemize}
  \item $\cA$ is good to guess $a$, i.e. $P_1 = \Pr[ \tilde{a} = a \mid (\tilde{a}, \tilde{b}) \leftarrow \cA(f(a,b))] \geq 1/2 + 1/\poly$
  \item $\cA$ is good to guess $b$, i.e. $P_2 = \Pr[ \tilde{b} = b \mid (\tilde{a}, \tilde{b}) \leftarrow \cA(f(a,b))] \geq 1/2 + 1/\poly$
  \item $\cA$ is good to guess the XOR of $a$ and $b$, i.e. $P_{\oplus} = \Pr[ \tilde{a} \xor \tilde{b} = a \xor b \mid (\tilde{a}, \tilde{b}) \leftarrow \cA(f(a,b))]  \geq 1/2 + 1/\poly$
  \end{itemize}
\end{lemma}

\begin{proof}
  Let's denote by:
  \begin{itemize}
  \item $e_1 = \pr{\tilde{a} \neq a \text{ and } \tilde{b} \neq b \mid (\tilde{a},\tilde{b}) \leftarrow \cA(f(a,b))}$
  \item $e_2 = \pr{\tilde{a} = a \text{ and } \tilde{b} \neq b \mid (\tilde{a},\tilde{b}) \leftarrow \cA(f(a,b))}$
  \item $e_3 = \pr{\tilde{a} \neq a \text{ and } \tilde{b} = b \mid (\tilde{a},\tilde{b}) \leftarrow \cA(f(a,b))}$
  \item $e_4 = \pr{\tilde{a} = a \text{ and } \tilde{b} = b \mid (\tilde{a},\tilde{b}) \leftarrow \cA(f(a,b))}$
  \end{itemize}
  Now, let us assume that the probability to do a correct guess is good, i.e. $e_4 \geq \frac{1}{4} + \frac{1}{\poly}$. Because the probability of guessing correctly $a$ (resp $b$) is bad, we have $e_2 + e_4 \leq \frac{1}{2} + \negl$ (resp. $e_3 + e_4 \leq \frac{1}{2} + \negl$), so $e_2 \leq \frac{1}{4} + \negl$ (resp. $e_3 \leq \frac{1}{4} + \negl$). So $e_2 + e_3 \leq \frac{1}{2} - \frac{1}{\poly}$, and because $e_1 + e_2 + e_3 + e_4 = 1$, we get $e_1 + e_4 \geq \frac{1}{2} + \frac{1}{\poly}$. But $e_1 + e_4$ is exactly the probability to guess the XOR, i.e.
  \[\pr{\tilde{a} \xor \tilde{b} = a \xor b \mid (\tilde{a},\tilde{b}) \leftarrow \cA(f(a,b))} \geq \frac{1}{2} + \frac{1}{\poly}\]

\end{proof}

\section{Replacing a quantum channel with verifiable QFactory and Verifiable Quantum Computation}\label{app:vQFactory_channel}

Here we prove that the verifiable QFactory can be used to replace \emph{any} protocol that has a quantum channel where the honest parties send random secret qubits of the form $\Ket{+_{L\pi/4}}$. We show this, with a simple reduction: if there exist a QPT adversary $\mathcal{A}$ that can break the protocol with the verifiable QFactory states $\rho_L$ with non-negligible probability $\frac{1}{p(n)}$, then there exist a QPT adversary $\mathcal{A}'$ that can break the security of the initial protocol that has quantum communication of the form $\Ket{+_{L\pi/4}}$, with the same probability of success. \\ \\
\procedure [linenumbering]{$\mathcal{A'}(\Ket{+_{L\pi/4}})$} {
 Prepare \, \, \sigma_{\textrm{junk}} \pccomment{junk does not depend on $\theta$} \\
 \gamma \gets  \Ket{+_{L\pi/4}}\Bra{+_{L\pi/4}} \otimes \sigma_{\textrm{junk}} \\
 \gamma' \gets \Phi^{-1}\left(\Ket{+_{L\pi/4}}\Bra{+_{L\pi/4}} \otimes \sigma_{\textrm{junk}}\right) =  
 \rho_L \pccomment{$\Phi$ is a QPT isometry}\\
 \pcreturn \mathcal{A}(\rho_L)\pccomment{succeeds with prob $\frac{1}{p(n)}$}
}

One of the most important possible applications of Verifiable QFactory is the classical client verifiable blind quantum computation. In particular, the verifiable blind quantum computation protocol of \cite{FKD2018} requires quantum communication in the beginning of the protocol and consists of strings of states of the form $\ket{+_{L\pi/4}}$ sent from the client to the server. According to our proof above, this means that the protocol of \cite{FKD2018} is as secure as a classical client protocol that replaces the quantum channel with verifiable QFactory.

\section{Strong Blindness}\label{app:strong_blindness}

We prove here that for our purposes, it is possible to slightly relax \autoref{def:verifiable_8_state}. In particular, given that the set of states $S = \{\rho_L \, | \, L \in \{0,1,...,7\}  \}$ is \textit{basis-blind}, and that there exists an index-independent isometry $\Phi$ that maps a state $\rho_L$ to a state $\epsilon$-close to $\ket{+_{L\pi/4}}\bra{+_{L\pi/4}}\otimes \sigma(L)$, then we can \emph{prove} that the junk state $\sigma(L)$ has to be independent of $L$ and thus satisfies Definition \ref{def:verifiable_8_state}. 

This property is (trivially) related with what we call ``\textit{strong blindness}'', since it essentially means that basis blindness (along with the existence of some isometry) guarantees that the only information that the server can learn from $L$ is exactly the information he can learn in an honest run. 

\begin{definition}[8-states Strong Blindness\label{def:strong_blindness}]
Consider a party that is given a state uniformly chosen from a set of eight states $S = \{\rho_L \, | \, L \in \{0,1, ..., 7\}\}$. We  say that $S$ is strongly blind if the information the party can learn about the index $L$ is bounded by the information that one can obtain from the state $\Ket{+_{L\pi/4}}$. 
\end{definition}

\begin{lemma}[8-states strong blindness from isometry and weak blindness\label{lemma:independent_junk}] Consider a party that is given a state uniformly chosen from a set of eight states $\{\rho_L\}$, where the set is basis blind. Assume also that there exists an isometry $\Phi$, that is independent of the index $L$, such that

\EQ{
\Phi(\rho_L)\overset{\epsilon}{\approx}\ket{+_{L\pi/4}}\bra{+_{L\pi/4}}\otimes \sigma(L)
}
where the state $\sigma(L)$ is a general junk state. Then we can show that $\sigma(L)$ should have negligible dependence on $L$, and therefore, by applying the inverse of the isometry, it satisfies 8-states strong blindness.
\label{lemma:strong_blind}
\end{lemma}

\begin{proof}
We prove this by contradiction. Assume that the state $\sigma(L)$ does not have negligible dependence on $L$. This means that one can guess $L$ with non-negligible advantage from a random guess, i.e. there exists some measurement such that when applied to $\sigma(L)$ the measurement outcome is $L$ with probability $Pr_{L}(\sigma(L)) = \frac{1}{8} + \frac{1}{poly(n)}$, for some polynomial $poly$. \\
From the 8 state basis blindness condition, we have that:
\EQ{
 \sigma(L) + \sigma(L \oplus 100) \overset{\epsilon}{\approx} \sigma(L') + \sigma(L' \oplus 100) \, \, \, \forall \  L, L' \in \{0,1,...,7\},
}
where $\oplus$ refers to bitwise xor. \\
Then, from this condition, we also deduce that: $Pr_{L \oplus 100}(\sigma(L)) = \frac{1}{8} - \frac{1}{poly(n)}$. \\
Now, consider the 2 bases: $\{L, L \oplus 100\}$ and $\{L \oplus 010, L \oplus 110\}$. We will show next how to construct a distinguisher between these 2 basis. The idea is to use the information about the index $L$ obtained from $\sigma(L)$ and then perform an optimal quantum measurement on the $\ket{+_{L\pi/4}}$ state given the prior knowledge that we obtained from $\sigma(L)$. This will lead to guessing information about the basis bit with non-negligible probability.

Without loss of generality, we suppose $Pr_{L \oplus 010} \geq Pr_{L \oplus 110}$. \\
Then we get: 
\EQ{
 Pr \left(L \,  or \,  L \oplus 010 \right) = Pr_{L} +  Pr_{L \oplus 010} = Pr_{L \oplus 100} + \frac{2}{poly(n)} +  Pr_{L \oplus 010} \geq \nonumber \\
 Pr_{L \oplus 100} + Pr_{L \oplus 110} + \frac{2}{poly(n)} = Pr \left(L \oplus 100 \,  or \,  L \oplus 110 \right)
}
Then, to distinguish between the 2 bases, we will make a measurement on the $\ket{+_{L\pi/4}}$ state (first register), in the basis $\{L \oplus 011, L \oplus 111\}$. The resulting distinguishing probability is:
\EQ{
Pr_{success} &=& \frac{2 + \sqrt{2}}{4} \cdot \left( \frac{1}{2} + \frac{1}{poly(n)} \right) + \left(1 - \frac{2 + \sqrt{2}}{4} \right) \cdot \left(\frac{1}{2} - \frac{1}{poly(n)} \right) \nonumber \\
&=& \frac{1}{2} + \frac{1}{poly'(n)}
}
which, given the basis blindness assumption, reaches a contradiction.
\end{proof}

\section{Blind self-testing intermediate scenarios}\label{app:blind_self_testing_general}

Here we introduce a number of scenarios that takes us from a very simple setting for blind self-testing to the general case of \autoref{protocol:general_blind_self-testing}.

\begin{algorithm}[H]
\caption{Blind Self-Testing: The independent identically distributed case (Scenario 1)}\label{protocol:iid_blind_self-testing}

-- The server chooses eight states $\{\ket{\Psi(000)},\ket{\Psi(001)},\cdots, \ket{\Psi(111)}\}$, such that they are basis blind\footnotemark , i.e.  

\EQ{\label{eq:blindness_self_testing}
\Delta(\rho_{L}+\rho_{L \oplus 100},\rho_{L'}+\rho_{L' \oplus 100})\leq\epsilon_1 \ \ \forall \ L,L' \in \{0,1,...,7\}
}
where $\rho_L:=\ket{\Psi(L)}\bra{\Psi(L)} \, \, \forall  L \in \{0,1,...,7\}$. \\
-- The client chooses randomly $N$ indices, $\{ L_1,\cdots,L_N \}$, where each $L_i \leftarrow \{000,\cdots, 111\}$ and sends the set of states $\{\ket{\Psi(L_1)},\ket{\Psi(L_2)},\cdots, \ket{\Psi(L_N)}\}$ to the server, while keeping the indices $L_i$ secret.\\
-- The client, randomly chooses a fraction $f$ of the qubits to be used as tests and announces the set of corresponding indices $T=\{i_1,\cdots,i_{fN}\} \subset \{1,2,...,N\}$ to the server.\\
-- For each test qubit $i_j \in T$, the client chooses a random measurement index $M_{i_j}\in\{000,\cdots,111\}$ and instructs the server to measure the corresponding qubit in the $\left\{\ket{+_{M_{i_j}\pi/4}},\ket{-_{M_{i_j}\pi/4}}\right\}$ basis. \\
-- The server returns the test measurement results $\{c_{(i_j)}\}$. \\
-- For each fixed pair $(L,M)$, the client gathers all the test positions that correspond to that pair and 
from the relative frequencies, the client obtains an estimate for the probability $p_{L,M}$ (where by convention we have that $p_{L,M}$ corresponds to the $+1$ outcome, while $1-p_{L,M}$ to the $-1$).\\
-- If, for any pair $(L,M)$ :

\EQ{\label{eq:statistics_iid_self_testing}
|p_{L,M}-\cos^2((L-M)\pi/8)|\geq \epsilon_2
}
the client aborts.\\
\textbf{Output:} If the client does not abort (and this happens with non-negligible probability), then there exist an index-independent isometry $\Phi$, given below in Fig. \ref{fig:isometry_circuit}, such that

\EQ{\label{eq:isometry_iid_self_testing}
\Phi(\ket{\Psi(L)})\overset{\epsilon(\epsilon_1,\epsilon_2)}{\approx} \ket{+_{L\pi/4}}\otimes\ket{\Psi(000)}
}
\end{algorithm}

\footnotetext{The distance represents either trace-distance in the information theoretic security setting or QPT distinguisher in the computational security setting.}

In Scenario 1, we make a number of assumptions. Firstly, the server actually knows the classical description of all eight possible states. This is done to simplify things (see scenario 2b), but has one disadvantage. The server can compute $\rho_L+\rho_{L \oplus 100}$ and unless $\Delta(\rho_{L}+\rho_{L \oplus 100},\rho_{L'}+\rho_{L' \oplus 100})\leq\epsilon_1$ where $\Delta$ is the trace distance, the server could  guess the basis (with non-negligible advantage) by performing a minimum error measurement between these two mixed states. The measurement itself is likely to be of polynomial complexity, therefore it seems impossible to guarantee computational basis blindness in this setting (unless, of course, the states are also information theoretically close). For this reason, and to simplify the first exposition to \textit{blind self-testing}, we give the proof by assuming that Eq. \ref{eq:blindness_self_testing} is close in trace-distance in \autoref{app:blind_self-testing_iid}. Now we proceed with the other scenarios of blind self-testing, so that the relevance of this notion for verifiable QFactory becomes apparent.

\begin{algorithm}[H]
\caption{Blind self-testing: The independent non-identically distributed case (Scenario 2a)}\label{protocol:inid_blind_self-testing}

-- The server chooses $N$ eight-plets of states $\{\ket{\Psi(000)}_i,\ket{\Psi(001)}_i,\cdots, \ket{\Psi(111)}_i\}_{i\in\{1,\cdots,N\}}$, such that they are basis blind, i.e.  

\EQ{
\Delta(\rho_{i,L}+\rho_{i,L \oplus 100},\rho_{i,L'}+\rho_{i,L' \oplus 100})\leq\epsilon_1 \ \ \forall \ i \in \{1, ..., N\} \, \, \forall \ L,L' \in \{0,...,7\}
}
where $\rho_{i,L}:=\ket{\Psi(L)}_i\bra{\Psi(L)}_i$.\\
-- The client chooses randomly $N$ indices, $\{ L_1,\cdots,L_N \}$, where each $L_i \leftarrow \{000,\cdots, 111\}$ and sends the set of states $\{\ket{\Psi(L_1)}_1,\ket{\Psi(L_2)}_2,\cdots, \ket{\Psi(L_N)}_N\}$ to the server, while keeping the indices $L_i$ secret.\\
-- The client, randomly chooses a fraction $f$ of the qubits to be used as tests and announces the set of corresponding indices $T=\{i_1,\cdots,i_{fN}\} \subset \{1,2,...,N\}$ to the server. \\
-- For each test qubit $i_j \in T$, the client chooses a random measurement index $M_{i_j}\in\{000,\cdots,111\}$ and instructs the server to measure the corresponding qubit in the $\left\{\ket{+_{M_{i_j}\pi/4}},\ket{-_{M_{i_j}\pi/4}}\right\}$ basis. \\
-- The server returns the test measurement results $\{c_{(i_j)}\}$. \\
-- For each fixed pair $(L,M)$, the client gathers all the test positions that correspond to that pair and 
from the relative frequencies, the client obtains an estimate for the probability $p_{L,M}$ (where by convention we have that $p_{L,M}$ corresponds to the $+1$ outcome, while $1-p_{L,M}$ to the $-1$). Note, that each $p_{L,M}$ involves (in general) the statistics from different states.\\
-- If $|p_{L,M}-\cos^2((L-M)\pi/8)|\geq \epsilon_2$ for any pair $(L,M)$ the client aborts.\\
\textbf{Output:} If the client does not abort (and this happens with non-negligible probability), then there exists an index-independent isometry $\Phi$, given below in Fig. \ref{fig:isometry_circuit}, that if applied to a random non-tested qubit $i$, is acting in the following way:

\EQ{
\Phi(\ket{\Psi(L)}_i)\overset{\epsilon(\epsilon_1,\epsilon_2)}{\approx} \ket{+_{L\pi/4}}\otimes\ket{\Psi(000)}_i
}
\end{algorithm}

We note that Scenario 2a is similar with scenario 1 with the crucial difference that different sets of eight states are used for each of the $N$ qubits. It is not hard to see that very similar analysis with the one of Scenario 1 will apply, if instead of $N$ different sets, one has $N$ copies of the same state, but replaces that state with the average state defined to be $\rho_L(average):=\sum_{i=1}^N\rho_{i,L}$. The result will then hold with high probability using Hoeffding inequalities \cite{Hoeffding}.

\begin{algorithm}[H]
\caption{Blind self-testing: The independent non-identically distributed case (Scenario 2b)}\label{protocol:inid_blind_self-testing2}
-- The server prepares $N$ states $\ket{\Psi(L)}_{i} , i \in\{1,\cdots,N\}, L \in \{0, ..., 7\}$. For each of these states, the server is basis blind, i.e. (being computationally bounded) he cannot determine the two basis bits of the index $L$. On the other hand, the client does know the index $L$ for each of the $N$ states.\footnotemark \\
-- The client, randomly chooses a fraction $f$ of the qubits to be used as tests and announces the set of corresponding indices $T=\{i_1,\cdots,i_{fN}\} \subset \{1,2,...,N\}$ to the server. \\
-- For each test qubit $i_j \in T$, the client chooses a random measurement index $M_{i_j}\in\{000,\cdots,111\}$ and instructs the server to measure the corresponding qubit in the $\left\{\ket{+_{M_{i_j}\pi/4}},\ket{-_{M_{i_j}\pi/4}}\right\}$ basis. \\
-- The server returns the test measurement results $\{c_{(i_j)}\}$. \\
--- For each fixed pair $(L,M)$, the client gathers all the test positions that correspond to that pair and 
from the relative frequencies, the client obtains an estimate for the probability $p_{L,M}$ (where by convention we have that $p_{L,M}$ corresponds to the $+1$ outcome, while $1-p_{L,M}$ to the $-1$). Note, that each $p_{L,M}$ involves (in general) the statistics from different states.\\
-- If $|p_{L,M}-\cos^2((L-M)\pi/8)|\geq \epsilon_2$ for any pair $(L,M)$ the client aborts.\\
\textbf{Output:} If the client does not abort (and this happens with non-negligible probability), then there exist an index-independent polynomial isometry $\Phi$, given below in Fig. \ref{fig:isometry_circuit}, that if applied to a random non-tested qubit $i$ is acting in the following way:
\EQ{
\Phi(\ket{\Psi(L)}_i)\overset{\epsilon(\epsilon_1,\epsilon_2)}{\approx} \ket{+_{L\pi/4}}\otimes\ket{\Psi(000)}_i
}
\end{algorithm}
\footnotetext{This can be achieved either by having the client have exponential capacity, or by having the server prepare the states using some choice made by the client that has kept some side trapdoor information too.}

The crucial difference in scenario 2b, is that the server prepares only one state (not eight). In the previous scenarios, it was the client choosing which index $L$ is used, and thus, unless the states $\ket{\Psi(L)}$ leaked information about the (basis bits of the) index, the server was blind. Here we impose this by requiring explicitly that the server prepares a state that he is basis-blind with respect to its index. There are two consequences of these differences.\\
First, now that the state is prepared on the server side, the client does not need to have any quantum ability, and his part in the protocol is purely classical. Second, since the state is prepared in the server's side, it is clear that we can no longer be in the information-theoretic setting, since, with unbounded computation power, he would be able to recover the exact label $L$. On the positive side, the issue we had in scenario 1, that the server could perform minimum error measurements is no longer valid, since the server does not know the classical description of the two states ($\rho_{i,L}+\rho_{i,L \oplus 100};\rho_{i,L'}+\rho_{i,L' \oplus 100}$) that he needs to distinguish and thus cannot find the corresponding minimum error measurement. Finally, some care is needed to specify how the client can possibly know the index $L$ while the server (that prepares the state) he does not. One way to achieve this is given in Malicious 8-states QFactory. Actually, scenario 2b corresponds exactly to the setting of Verifiable QFactory, provided that whatever deviation the server does in one round of Malicious 8-states QFactory is restricted (and independent) to other rounds. 

Finally, we can generalise further (\autoref{protocol:general_blind_self-testing} removing the assumption of tensor produce states.

\section{Proof of Scenario 1: i.i.d. blind self-testing}\label{app:blind_self-testing_iid}

Given the protocol \ref{protocol:iid_blind_self-testing} (Scenario 1), we can assume that there exist eight untrusted, binary observables $O'_{M}$, with $\pm1$ eigenvalues, where we define $O'_0:=X',O'_{010}:=Y'$. Each of these observable are of polynomial size, i.e. can be performed by a QPT party. The corresponding ideal observables are denoted without prime and we have $O_M=\ket{+_{M\pi/4}}\bra{+_{M\pi/4}}-\ket{-_{M\pi/4}}\bra{-_{M\pi/4}}$. We denote $\rho_L$ the set of the eight (untrusted) states, and as defined in scenario 1 we consider the pure states (a purification in general) $\ket{\Psi(L)}$.

We can see that index independent isometry in \autoref{fig:isometry_circuit} gives us:

\EQ{\label{eq:isometry0}\Phi(\ket{+}\ket{\Psi(L)}):=\ket{+}\frac{(I+X')}{2}\ket{\Psi(L)}+\ket{-}(-iY')\frac{(I-X')}{2}\ket{\Psi(L)}
}
and we want to show that the state in Eq. (\ref{eq:isometry0}) is $\overset{\epsilon}{\approx} \ket{+_{L\pi/4}}\otimes \ket{\Psi(000)}$ using the constrains coming from the basis blindness property and the measurement statistics of the test qubits. In this first exposition to the blind self-testing, we are going to prove the correctness up to $\epsilon$, while the details of the robustness of this result (the explicit dependence of the final $\epsilon$ on $\epsilon_1,\epsilon_2$ appearing in the constraints) is left for future work. 

\begin{theorem}\label{thm:iid_self_testing}
If Protocol \ref{protocol:iid_blind_self-testing} does not abort (with non-negligible probability) then the isometry given by Fig. \ref{fig:isometry_circuit} satisfies the condition of Eq. (\ref{eq:isometry_iid_self_testing}).
\end{theorem}
To prove this theorem we first need two lemmas.

\begin{lemma}\label{lemma:2d_subspace}
For any set of (eight) states $\{\ket{\Psi(L)}\}$ as from scenario 1, that is basis blind, i.e. $|(\rho_{L}+\rho_{L \oplus 100})-(\rho_{L'}+\rho_{L' \oplus 100})|\leq \epsilon_1$, all eight states belong (approximately) in a 2-dimensional subspace spanned by the vectors $\{\ket{\Psi(000)},\ket{\Psi(100)}\}$.
\end{lemma}

\begin{proof}
The binary observables $O'_M$ have $\pm 1$ eigenvalues, and can therefore be written as $O'_M=(+1)P^+_{M}+(-1)P^-_{M}$ where $P_M^{\pm}$ are the projections on the $+1$ and $-1$ eigenspace, respectively. It follows trivially that $(O'_M)^2=I$ and that the corresponding projections are $P^+_M=\frac{I+O_M'}{2}$ and $P^-_M=\frac{I-O_M'}{2}$.

Moreover, if the protocol does not abort (and this happens with non-negligible probability), we also have the constraints on the expectation values of the observables coming from Eq. (\ref{eq:statistics_iid_self_testing}). From $\bra{\Psi(L)}O'_L\ket{\Psi(L)}\overset{\epsilon_2}{\approx}1$ and $\bra{\Psi(L)}O'_{L \oplus 100}\ket{\Psi(L)}\overset{\epsilon_2}{\approx}-1$, we get that $O'_L\ket{\Psi(L \oplus 100)}\approx -\ket{\Psi(L \oplus 100)}$ and thus

\EQ{P^+_L\ket{\Psi(L)}\approx\ket{\Psi(L)} &,& P^-_L\ket{\Psi(L)}\approx 0 \nonumber\\  
P^-_L\ket{\Psi(L \oplus 100)}\approx\ket{\Psi(L \oplus 100)} &,& P^+_L\ket{\Psi(L \oplus 100)}\approx 0\nonumber
}  
which means that $\bra{\Psi(L \oplus 100)}\Psi(L)\rangle\approx 0$. The space spanned by two vectors $\{\ket{\Psi(L)},\ket{\Psi(L \oplus 100)}\}$ is two dimensional and for all $L$ the state $1/2(\rho_L+\rho_{L \oplus 100})$ is the identity in that subspace. Now from the basis blindness condition we have that $1/2(\rho_L+\rho_{L \oplus 100})\overset{\epsilon_1}{\approx}1/2(\rho_{000}+\rho_{100})$, which means that all states $\ket{\Psi(L)}$ belong to that (fixed) 2-dimensional subspace. We will denote the projection on this 2-dimensional subspace as $P_{\Psi}$.
\end{proof}

\begin{lemma}\label{lemma:state_iid}
The states given in Protocol \ref{protocol:iid_blind_self-testing} are approximately of the form 
\EQ{\label{eq:general_state_iid_case}
\ket{\Psi(L)}\overset{\epsilon}{\approx} \cos(\frac{L\pi}{8})\ket{\Psi(000)}+e^{i\phi}\sin(\frac{L\pi}{8})\ket{\Psi(100)}
}
with $\phi$ being a constant (independent of $L$). Furthermore, the untrusted operator $Y'$ acts in the following way:
\EQ{\label{eq:Y_action}
Y'\ket{\Psi(100)}=e^{-i\phi}\ket{\Psi(000)}
}
\end{lemma}

\begin{proof}
From Lemma \ref{lemma:2d_subspace} we know we can express all states in the following form:

\EQ{
\ket{\Psi(L)}\approx e^{if_2(L)}(a_0(L)\ket{\Psi(000)}+e^{if_1(L)}a_1(L)\ket{\Psi(100)})
}
where $a_0(L),a_1(L), f_1(L)$ are functions to be determined and $f_2(L)$ is an overall complex phase that we could ignore, but we keep it here to remind that we can use this to simplify the final expressions. From the statistics of the measurements of the $X'$ observable, we get (directly) that $a_0(L)\approx|\cos\frac{L\pi}{8}|\ , \ a_1(L)\approx|\sin\frac{L\pi}{8}|$.

\EQ{
\ket{\Psi(L)}\approx |\cos\frac{L\pi}{8}|\ket{\Psi(000)}+e^{if_1(L)}|\sin\frac{L\pi}{8}|\ket{\Psi(100)}
}
where we dropped the global phase $e^{if_2(L)}$. Now for $L\neq 000,100$, from $\bra{\Psi(L)}\Psi(L \oplus 100)\rangle\approx 0$ we obtain that:

\EQ{\label{eq:f1_1}
f_1(L \oplus 100)=(f_1(L)+\pi) \ \ \bmod 2\pi
}
and we can express the states grouped in four orthogonal bases:

\EQ{\label{eq:f1_0} \ket{\Psi(000)};\\
\ket{\Psi(001)} &=& |\cos\frac{\pi}{8}|\ket{\Psi(000)}+e^{if_1(001)}|\sin\frac{\pi}{8}|\ket{\Psi(100)}  ;  \nonumber\\
\ket{\Psi(010)} &=& \frac{1}{\sqrt{2}}\left(\ket{\Psi(000)}+e^{if_1(010)}\ket{\Psi(100)}\right);  \nonumber\\
\ket{\Psi(011)} &=& |\sin\frac{\pi}{8}|\ket{\Psi(000)}+e^{if_1(011)}|\cos\frac{\pi}{8}|\ket{\Psi(100)} ; \nonumber \\
 \ket{\Psi(100)}; \nonumber\\
 \ket{\Psi(101)} &=& |\sin\frac{\pi}{8}|\ket{\Psi(000)}-e^{if_1(001)}|\cos\frac{\pi}{8}|\ket{\Psi(100)}; \nonumber \\
 \ket{\Psi(110)} &=& \frac{1}{\sqrt{2}}\left(\ket{\Psi(000)}-e^{if_1(010)}\ket{\Psi(100)}\right) ; \nonumber \\
 \ket{\Psi(111)} &=& |\cos\frac{\pi}{8}|\ket{\Psi(000)}-e^{if_1(011)}|\sin\frac{\pi}{8}|\ket{\Psi(100)} \nonumber \\
}
where we used  some identities such as $|\cos(3\pi/8)|=|\sin(\pi/8)|$, etc. Now, we have three parameters to fix, namely $f_1(001),f_1(010),f_1(011)$. For notational simplicity, we will use $c:=|\cos\pi/8| \ ; \ s:= |\sin \pi/8|$.

Then, we will use the statistics that we have from Eq. (\ref{eq:statistics_iid_self_testing}), when measuring in a different than the $X'$ basis. Expressing $\ket{\Psi(001)}$ in the $Y'$ basis we get:

\EQ{
\ket{\Psi(001)}=\frac{1}{\sqrt{2}}\left(c+s e^{i(f_1(001)-f_1(010))}\right)\ket{\Psi(010)}+ \frac{1}{\sqrt{2}}\left(c-s e^{i(f_1(001)-f_1(010))}\right)\ket{\Psi(110)}
}
From Eq. (\ref{eq:statistics_iid_self_testing}) and the probability of obtaining the result $010$ when having the state $\ket{\Psi(001)}$ we obtain:

\EQ{
\frac{1}{\sqrt{2}}|\left(c+s e^{i(f(001)-f(010))}\right)|\approx c
}
that is possible only if $e^{i(f(001)-f(010))}=1$, i.e. we have:

\EQ{\label{eq:f1_2}
f_1(001)=f_1(010) \ \ \bmod 2\pi
}

Similarly, by expressing $\ket{\Psi(011)}$ in the $Y'$ basis we get:

\EQ{
\ket{\Psi(011)}= \frac{1}{\sqrt{2}}\left(s+c e^{i(f_1(011)-f_1(010))}\right)\ket{\Psi(010)}+ \frac{1}{\sqrt{2}}\left(s-c e^{i(f_1(011)-f_1(010))}\right)\ket{\Psi(110)}
}
From Eq. (\ref{eq:statistics_iid_self_testing}) and the probability of obtaining the result $010$ when having the state $\ket{\Psi(011)}$, we have:

\EQ{
\frac{1}{\sqrt{2}}|\left(s+c e^{i(f(011)-f(010))}\right)|\approx c 
}
that is possible only if $e^{i(f(011)-f(010))}=1$, i.e. we get:

\EQ{\label{eq:f1_3}
f_1(011)=f_1(010) \ \ \bmod 2\pi
}
Setting $f(001):=\phi$, we use Eqs. (\ref{eq:f1_1}, \ref{eq:f1_2}, \ref{eq:f1_3}) and Eq. (\ref{eq:f1_0}) becomes:

\EQ{
\ket{\Psi(000)};  \\
\ket{\Psi(001)} &=& \cos\frac{\pi}{8}\ket{\Psi(000)}+e^{i\phi}\sin\frac{\pi}{8}\ket{\Psi(100)};  \nonumber\\
\ket{\Psi(010)}&=& \cos\frac{2\pi}{8}\ket{\Psi(000)}+e^{i\phi}\sin\frac{2\pi}{8}\ket{\Psi(100)};  \nonumber\\
\ket{\Psi(011)}&=& \cos\frac{3\pi}{8}\ket{\Psi(000)}+e^{i\phi}\sin\frac{3\pi}{8}|\ket{\Psi(100)}; \nonumber\\
\ket{\Psi(100)}; \nonumber\\
\ket{\Psi(101)} &=& -\cos\frac{5\pi}{8}\ket{\Psi(000)}-e^{i\phi}\sin\frac{5\pi}{8}\ket{\Psi(100)}; \nonumber\\
\ket{\Psi(110)}&=&-\cos\frac{6\pi}{8}\ket{\Psi(000)}-e^{i\phi}\sin\frac{6\pi}{8}\ket{\Psi(100)} ; \nonumber\\
 \ket{\Psi(111)} &=& -\cos\frac{7\pi}{8}\ket{\Psi(000)}-e^{i\phi}\sin\frac{7\pi}{8}\ket{\Psi(100)}; \nonumber\\
}
where we used $\cos\frac{5\pi}{8}=-\sin\frac{\pi}{8}$; $\sin\frac{5\pi}{8}=\cos\frac{\pi}{8}$;  $-\cos\frac{6\pi}{8}=\frac{1}{\sqrt{2}}$; $\sin\frac{6\pi}{8}=\frac{1}{\sqrt{2}}$; $\cos\frac{3\pi}{8}=\sin\frac{\pi}{8}$; $\sin\frac{3\pi}{8}=\cos\frac{\pi}{8}$; $-\cos\frac{7\pi}{8}=\cos\frac{\pi}{8}$ and $\sin\frac{7\pi}{8}=\sin\frac{\pi}{8}$.  By noting that each state is invariant if multiplied by a global phase (an overall minus sign in our case) we get the expression:

\EQ{
\ket{\Psi(L)}= \cos(\frac{L\pi}{8})\ket{\Psi(000)}+e^{i\phi}\sin(\frac{L\pi}{8})\ket{\Psi(100)}
}
as required for Lemma \ref{lemma:state_iid}. The action of the observable $Y'$ projected on the subspace that all $\ket{\Psi(L)}$ belong to, is:

\EQ{
P_{\Psi}Y'P_{\Psi}= \ket{\Psi(010)}\bra{\Psi(010)}-\ket{\Psi(110)}\bra{\Psi(110)}
}
and given that 

\EQ{
\ket{\Psi(010)} &=& \frac{1}{\sqrt{2}}\left(\ket{\Psi(000)}+e^{i\phi}\ket{\Psi(100)}\right)\nonumber\\
\ket{\Psi(110)} &=& \frac{1}{\sqrt{2}}\left(-\ket{\Psi(000)}+e^{i\phi}\ket{\Psi(100)}\right)
}
we get:

\EQ{
P_\Psi Y' P_\Psi &=& e^{-i\phi}\ket{\Psi(000)}\bra{\Psi(100)}+e^{i\phi}\ket{\Psi(100)}\bra{\Psi(000)}\nonumber\\
Y'\ket{\Psi(100)}&=& e^{-i\phi}\ket{\Psi(000)}}
p\end{proof}

\begin{proof}[Proof of Theorem \ref{thm:iid_self_testing}]
We can now prove Theorem \ref{thm:iid_self_testing}. If the protocol does not abort, it means that the condition on the test measurement statistics is satisfied, and the above Lemmas hold. We substitute Eq. (\ref{eq:general_state_iid_case}) in the isometry, note that $\frac{(I+X')}{2}\ket{\Psi(000)}=1=\frac{(I-X')}{2}\ket{\Psi(100)}$ and $\frac{(I+X')}{2}\ket{\Psi(100)}=0=\frac{(I-X')}{2}\ket{\Psi(000)}$ and use Eq. (\ref{eq:Y_action}) to get the desired result:

\EQ{\Phi(\ket{+}\ket{\Psi(L)})&=&\ket{+}\frac{(I+X')}{2}\ket{\Psi(L)}+\ket{-}(-iY')\frac{(I-X')}{2}\ket{\Psi(L)}\nonumber\\
&=& \ket{+}\cos(\frac{L\pi}{8})\ket{\Psi(000)}+\ket{-}(-iY')e^{i\phi}\sin(\frac{L\pi}{8})\ket{\Psi(100)}\nonumber\\
&=&\ket{+}\cos(\frac{L\pi}{8})\ket{\Psi(000)}+\ket{-}e^{i\phi}\sin(\frac{L\pi}{8})(-ie^{-i\phi})\ket{\Psi(000)}\nonumber\\
&=&\left(\cos(\frac{L\pi}{8})\ket{+}-i\sin(\frac{L\pi}{8})\ket{-}\right)\ket{\Psi(000)}\nonumber\\
&=&\ket{+_{L\pi/4}}\ket{\Psi(000)}
}
\end{proof}

\section{Generalisation to pseudo-homomorphic functions \label{app:pseudohom}}

\tocless\subsection{Notation and remarks about the generalisation}

Note that later on, we will do a slight abuse of notation, and if $\cK$ is a set of keys, we will denote by $k \gets \cK$ the sampling of a key in $\cK$. Note that this sampling may not be uniform, and depends on a fixed distribution.

Moreover, in order to have a function usable in practice, we would like to be able to create the uniform superposition on all the elements of the input set. However, in practice, it may not be possible to have an exact superposition on all the elements of this input set. Therefore we will consider in the following that the function $g$ has an input set $\cZ$ that is a bit bigger, but so that we can create an exact uniform superposition on all the elements of this set, and in order to deal with the fact that $\cZ$ is not the initial input set, we will say that $g(z) = \bot$ as soon as $z$ does not belong to the initial input set of $g$. Note that we may also abuse notation for simplicity, and also write $g(z) = \bot$ when $z$ does not even belong to $\cZ$ and when $g$ is not defined on a given input $z$.

We will also need to extend some notions to this new notion. For example, we will say that a function $g: \cZ \rightarrow \cY \cup \bot$ is \textbf{injective} when for all $y \in \cY$, $|f^{-1}(y)| \leq 1$.

\tocless\subsection{Definition}

\begin{definition}[($\eta$, $\cZ$, $\cZ_0$, $\cD$)-homomorphic family of functions]
  Let us consider a family of functions $\{g_k\colon \cZ \rightarrow \cY \cup \bot\}_{k \in \cK}$, as well as two symmetric binary group relations $\ast$ and $\star$, with $\ast$ acting on a set containing $\cZ$ and $\cZ_0$, $\star$ acting on $\cY \cup \bot$, and so that $\forall y \in \cY, \bot \star y = \bot$. We say that $\{f_k\}_{k \in \cK}$ is an ($\eta$, $\cZ$, $\cZ_0$, $\cD$)-homomorphic function if $\cD$ is a distribution on $\cZ_0$ and \[\Pr_{\substack{k \gets \cK\\ z_0 \gets^\cD \cZ_0\\ z \sample \cZ}}[z \ast z_0 \in \cZ \text{ and } g_k(z) \star g_k(z_0) = g_k(z \ast z_0) \neq \bot] \geq \eta \] Note that we do require that $z$ is sampled uniformly from $\cZ$, but $z_0$ is sampled from a distribution $\cD$ on $\cZ_0$ that may not be uniform.
\end{definition}

\begin{definition}[$\delta$-2-regular family of functions]
  Let us consider a family of functions $\{f_k\colon \cX \rightarrow \cY \cup \bot\}_{k \in \cK}$. For a fixed $k$, $\cY^{(2)}$ will be the set of $y$ having two preimages: $\cY_{f_k}^{(2)} = \{y \in \cY, |f_k^{-1}(y)| = 2\}$. Then, this family of functions is said to be $\delta$-2-regular if \[\Pr_{\substack{k \sample \cK \\ x \sample \cX}}[f_k(x) \in \cY_{f_k}^{(2)}] \geq \delta\]
\end{definition}

\begin{lemma}[($\eta$, $\cZ$, $\cZ_0$)-homomorphy to $\delta$-2-regularity]
  Given a family of functions $\{g_k\colon \cZ \rightarrow \cY \cup \bot\}_{k \in \cK}$ that is both injective and an ($\eta$, $\cZ$,  $\cZ_0$)-homomorphic family of functions, then it's possible to build a family $\{f_{k'}\colon \cZ \times \{0,1\} \rightarrow \cY \cup \bot\}_{k' \in \cK'}$ that is $\delta$-2-regular, with $\delta = \eta$.
  \label{lemma:approx_two_reg}
\end{lemma}
\begin{proof}
  Let's do the following construction. To sample a key $k' \in \cK'$, we first sample a key $k$ from $\cK$, as well as an $z_0 \gets^{\cD} \cZ_0$, and we define $k' = (k, y_0 := f_k(z_0))$. Then, we define $f_{k'}(z,0) = g_k(z)$ and $f_{k'}(z,1) = g_k(z) \star y_0$, also denoted later as $f_{k'}(z,c) = g_k(z) \star (c \cdot y_0)$ for simplicity. Now, we remark that:
  \begin{align}
   &\Pr_{\substack{k' \sample \cK' \\ x \sample \cX}}[f_{k'}(x) \in \cY_{f_k'}^{(2)}] \\
   =& \Pr_{\substack{k \sample \cK \\ z_0 \gets^{\cD} \cZ_0 \\ z \sample \cZ
   \\ c \sample \{0,1\}}}[g_k(z) \star (c \cdot g_k(z_0)) \in
   \cY_{f_k'}^{(2)}] \\
    =& \frac{1}{2} \times \Big( \Pr_{\substack{k \sample \cK \\ z_0 \gets^{\cD} \cZ_0 \\ z \sample \cZ}}[g_k(z) \in \cY_{f_k'}^{(2)}]
    + \Pr_{\substack{k \sample \cK \\ z_0 \gets^{\cD} \cZ_0 \\ z \sample \cZ}}[g_k(z) \star g_k(z_0) \in \cY_{f_k'}^{(2)}]  \Big)\\
    \geq & \frac{1}{2} \times \Big( \Pr_{\substack{k \sample \cK \\ z_0 \gets^{\cD} \cZ_0 \\ z \sample \cZ}}[g_k(z) \in \cY_{f_k'}^{(2)} \text{ and } \underbrace{z \ast z_0^{-1} \in \cZ \text{
      and } g_k(z) = g_k(z \ast z_0^{-1}) \star g_k(z_0) \neq \bot}_{C_1}]\\
    &+ \Pr_{\substack{k \sample \cK \\ z_0 \gets^{\cD} \cZ_0 \\ z \sample \cZ}}[g_k(z) \star g_k(z_0) \in \cY_{f_k'}^{(2)} \text{ and } \underbrace{z \ast z_0 \in \cZ \text{
      and } g_k(z) \star g_k(z_0) = g_k(z \ast z_0) \neq \bot}_{C_2}]  \Big)\\
    = & \frac{1}{2} \times \Big( \Pr_{\substack{k \sample \cK \\ z_0 \gets^{\cD} \cZ_0 \\ z \sample \cZ}}[g_k(z) \in \cY_{f_k'}^{(2)} | C_1 ] \times \Pr_{\substack{k \sample \cK \\ z_0 \gets^{\cD} \cZ_0 \\ z \sample \cZ}}[C_1]\\
    &+ \Pr_{\substack{k \sample \cK \\ z_0 \gets^{\cD} \cZ_0 \\ z \sample \cZ}}[g_k(z) \star g_k(z_0) \in \cY_{f_k'}^{(2)} | C_2] \times \Pr_{\substack{k \sample \cK \\ z_0 \gets^{\cD} \cZ_0 \\ z \sample \cZ}}[C_2]  \Big)
  \end{align}

   Now, remark that when $z_0 \ast z \in \cD$ and $g_k(z_0) \star g_k(z) = g_k(z_0 \ast z) \neq \bot$, then $y := g_k(z) \star g_k(z_0) \in \cY_{f_k'}^{(2)}$. Indeed:
   \begin{itemize}
   \item $y \in \cY$ because $g_k(z) \star g_k(z_0) \neq \bot$ and the $\star$ operator is defined on $\cY \cup \bot$
   \item there are at least two preimages mapping to $y$, because $y = f_k(z,1) = g_k(z) \star g_k(z_0) = g_k(z \ast z_0) = f_k(z \ast z_0, 0)$.
   \item there are at most two preimages mapping to $y$: indeed $g_k$ is injective, so both partial functions $f(\cdot, 0)$ and $f(\cdot, 1)$ are injective, so it's not possible to have more than two preimages mapping to $y$.
   \end{itemize}
   So $\Pr_{\substack{k \sample \cK \\ z_0 \gets^{\cD} \cZ_0 \\ z \sample \cZ}}[g_k(z) \star g_k(z_0) \in \cY_{f_k'}^{(2)} | C_2] = 1$. Similarly, $\Pr_{\substack{k \sample \cK \\ z_0 \gets^{\cD} \cZ_0 \\ z \sample \cZ}}[g_k(z) \in \cY_{f_k'}^{(2)} | C_1 ] = 1$.
   So we can rewrite the above equation as:
  \begin{align}
   &\Pr_{\substack{k' \sample \cK' \\ x \sample \cX}}[f_{k'}(x) \in \cY_{f_k'}^{(2)}] \\
    \geq & \frac{1}{2} \times \Big( \Pr_{\substack{k \sample \cK \\ z_0 \gets^{\cD} \cZ_0 \\ z \sample \cZ}}[C_1]
    + \Pr_{\substack{k \sample \cK \\ z_0 \gets^{\cD} \cZ_0 \\ z \sample \cZ}}[C_2]  \Big)
  \end{align}
  Now, remember that $\{g_k\}_k$ is ($\eta$, $\cZ$, $\cZ_0$)-homomorphic, so $\Pr_{\substack{k \sample \cK \\ z_0 \gets^{\cD} \cZ_0 \\ z \sample \cZ}}[C_2] \geq \eta$. By symmetry, we also have $\Pr_{\substack{k \sample \cK \\ z_0 \gets^{\cD} \cZ_0 \\ z \sample \cZ}}[C_1] \geq \eta$. Indeed:
  \begin{align}
    &\Pr_{\substack{k \sample \cK \\ z_0 \gets^{\cD} \cZ_0 \\ z \sample \cZ}}[\underbrace{z \ast z_0^{-1}}_{\hat{z}} \in \cZ \text{
    and } g_k(z) = g_k(z \ast z_0^{-1}) \star g_k(z_0) \neq \bot ]\\
    =&\Pr_{\substack{k \sample \cK \\ z_0 \gets^{\cD} \cZ_0 \\ z \sample \cZ}}[\hat{z} \in \cZ \text{ and } z \in \cZ \text{ and } \hat{z} = z \ast z_0^{-1} \text{ and } g_k(z) = g_k(\hat{z}) \star g_k(z_0) \neq \bot ]\\
    =&\Pr_{\substack{k \sample \cK \\ z_0 \gets^{\cD} \cZ_0 \\ \hat{z} \sample \cZ}}[\hat{z} \in \cZ \text{ and } z \in \cZ \text{ and } \hat{z} = z \ast z_0^{-1} \text{ and } g_k(z) = g_k(\hat{z}) \star g_k(z_0) \neq \bot ]\\
    =&\Pr_{\substack{k \sample \cK \\ z_0 \gets^{\cD} \cZ_0 \\ \hat{z} \sample \cZ}}[\hat{z} \ast z_0 \in \cZ\text{ and } g_k(\hat{z} \ast z_0) = g_k(\hat{z}) \star g_k(z_0) \neq \bot ]\\
    =&\Pr_{\substack{k \sample \cK \\ z_0 \gets^{\cD} \cZ_0 \\ z \sample \cZ}}[C_2]\\
    \geq& \eta
  \end{align}
  So $\Pr_{\substack{k' \sample \cK' \\ x \sample \cX}}[f_{k'}(x) \in \cY_{f_k'}^{(2)}] \geq \eta$, which concludes the proof.
\end{proof}

\section{Proof of the Malicious-Abort QFactory}\label{sec:proof_malicious_appendix}
In this section we moved some of the proofs of the security and correctness of \autoref{protocol:qfactory_abort_real}, which are summarised in \autoref{thm:malicious_abort_secure_correct}.

\printproofs

\section{Discussion about dealing with the abort case without Yao's XOR Lemma}\label{sec:discussion_qfactory_abort}

The proof of security we have when we cannot assume that abort arrives with negligible probability defined in \autoref{sec:qfactory_abort} has two drawbacks: first it relies on the conjecture that Yao's XOR Lemma is valid for one-round protocols (classical messages) with quantum adversary, but it also complicate the protocols by adding replication. We are also working on a second method that runs this time just one instance of Malicious 4-states QFactory, and will leaks the abort bit to the server. The specificity of this method is that the hash function will be send after receiving the $y$, and will be a 2-universal hash function (we just chose the function corresponding to the XOR of a random subset of hardcore bits that is easy to implement on server side as well). We also require that the $\delta$-2-regular trapdoor family of functions needs to have a polynomial number of homomorphic hardcore bits (we can easily extend our courstruction to have such requirements by just adding $q/2 \cdot (d_{0,1}, \dots, d_{0,t}, 0, \dots, 0)^T$). This setting will bring us back very close to the requirements needed by the leftover hash lemma \cite{Impagliazzo1989}. The advantage of this method is that we have a quantum equivalent of this lemma \cite{Tomamichel2010}, but unfortunately this lemma is valid (and expressed) in the information theoretic framework. Because our familly will never be information theoretically secure, but only computationally secure, we need an intermediate step to turn our information theoretically secure argument into a computationally secure one. Usually, this is done by introducing lossy functions \cite{PB2007}, and then using some indistinguishable property between injective and lossy functions to finish the proof. Unfortunately, our protocol has a non standard shape, and it's not yet clear how to define $a$ in the lossy case to make sure an adversary cannot exploit $a$ to distinguish between lossy and non lossy.

We also point out that we can also create some constructions where the abort bit is independent of the secret, by first sending a function whose goal is to create the superposition, and after receiving the $y$, if $y$ has two preimages, we send a single element with no noise $A s_0 + q/2 \cdot d_0$ that will be used like if it were another raw in the $y_0$. This problem is still supposed to be difficult, and because the second message does not have noise, it cannot lead to an abort, to the abort is independent of the secret. But the proof of security for this construction is also a work in progress.

\end{appendices}

\bibliographystyle{alpha}
\bibliography{biblio}

\end{document}